\documentclass[sigconf]{acmart}

\renewcommand\footnotetextcopyrightpermission[1]{}
\usepackage{bm}
\usepackage{tabularx}
\usepackage{booktabs}
\usepackage[subrefformat=parens]{subcaption}
\theoremstyle{definition}

\newcommand{\Nin}{\mathcal{N}^{\mathrm{in}}}
\newcommand{\Nout}{\mathcal{N}^{\mathrm{out}}}
\newcommand{\MIP}{P^\ast}
\newcommand{\MIOA}{MIOA_\theta}
\newcommand{\MIIA}{MIIA_\theta}
\newcommand{\TMP}{\beta_{\mathrm{tmp}}}
\newcommand{\TMPN}{\beta_{\mathrm{tmp}}^-}

\newcommand{\etal}{\textit{et al.,}}
\newcommand{\thead}[1]{\multicolumn{1}{c}{#1}}

\DeclareMathOperator*{\argmin}{arg\,min}
\DeclareMathOperator*{\argmax}{arg\,max}

\usepackage[linesnumbered,ruled,vlined]{algorithm2e}
\AtBeginDocument{%
  \providecommand\BibTeX{{%
    \normalfont B\kern-0.5em{\scshape i\kern-0.25em b}\kern-0.8em\TeX}}}

\begin{document}

\title{Network Prebunking Problem: Optimizing Prebunking Targets to Suppress the Spread of Misinformation in Social Networks}
\author{Satoshi Furutani}
\email{satoshi.furutani@ntt.com}
\orcid{1234-5678-9012}
\affiliation{%
  \institution{NTT Social Informatics Laboratories}
  \streetaddress{3-9-11, Midori-cho, Musashino-shi}
  \city{Tokyo}
  \country{Japan}
  \postcode{180-8585}
}

\author{Toshiki Shibahara}
\affiliation{%
  \institution{NTT Social Informatics Laboratories}
  \streetaddress{3-9-11, Midori-cho, Musashino-shi}
  \city{Tokyo}
  \country{Japan}
  \email{toshiki.shibahara@ntt.com}
}

\author{Mitsuaki Akiyama}
\affiliation{%
  \institution{NTT Social Informatics Laboratories}
  \streetaddress{3-9-11, Midori-cho, Musashino-shi}
  \city{Tokyo}
  \country{Japan}
  \email{mitsuaki.akiyama@ntt.com}
}

\author{Masaki Aida}
\affiliation{%
  \institution{Tokyo Metropolitan University}
  \streetaddress{6-6, Asahigaoka, Hino-shi}
  \city{Tokyo}
  \country{Japan}
  \email{aida@tmu.ac.jp}
}

\renewcommand{\shortauthors}{Furutani \etal}

\begin{abstract}
As a countermeasure against misinformation that undermines the healthy use of social media, a preventive intervention known as \textit{prebunking} has recently attracted attention in the field of psychology.
Prebunking aims to strengthen individuals' cognitive resistance to misinformation by presenting weakened doses of misinformation or by teaching common manipulation techniques before they encounter actual misinformation.
Despite the growing body of evidence supporting its effectiveness in reducing susceptibility to misinformation at the individual level, an important open question remains:
how best to identify the optimal targets for prebunking interventions to mitigate the spread of misinformation in a social network.
To address this issue, we formulate a combinatorial optimization problem, called the \textit{network prebunking problem}, which aims to select optimal prebunking targets that minimizes the spread of misinformation in a social network under limited intervention budgets.
We show that the problem is NP-hard and that its objective function is monotone and submodular, which provides a theoretical foundation for approximation guarantees of greedy algorithms. 
However, since the greedy algorithm is computationally expensive and does not scale to large networks, we propose an efficient approximation algorithm, MIA-NPP, based on the Maximum Influence Arborescence (MIA) approach, which restricts influence propagation around each node to a local directed tree rooted at that node. 
Through numerical experiments using real-world social network datasets, we demonstrate that MIA-NPP effectively suppresses the spread of misinformation under both fully observed and uncertain model parameter settings.
\end{abstract}



\keywords{Influence maximization, misinformation, prebunking, social network}



\maketitle
\pagestyle{plain}

\section{Introduction}

Social media has emerged as a major channel for information dissemination, yet it faces the serious threat of misinformation.
The spread of misinformation distorts elections by manipulating individual opinions and fueling political polarization, thereby undermining the integrity of democratic systems.
Moreover, during the COVID-19 pandemic, it triggered significant public health and social disruptions.

One widely adopted strategy to combat the threat of misinformation is \textit{debunking}, which involves retrospectively refuting false claims that have already spread by providing fact-checks and corrective information~\cite{smith2011correcting,lewandowsky2012misinformation}.
Numerous studies have demonstrated that debunking can be effective in reducing belief in misinformation and curbing its further spread~\cite{paynter2019evaluation, yousuf2021media, walter2018unring, chan2017debunking, martel2024fact}.
Traditionally, debunking is mainly conducted by media outlets and expert organizations. 
However, in recent years, social media platforms themselves have started to incorporate debunking features, as exemplified by Community Notes on X/Twitter~\cite{saeed2022crowdsourced, wojcik2022birdwatch}.
Despite their effectiveness and growing adoption, such post-hoc countermeasures face the challenge of removing the influence of misinformation from individuals once they have been exposed to it.
In fact, misinformation often continues to affect later reasoning even after it has been officially retracted or corrected, known as the continued influence effect~\cite{lewandowsky2012misinformation}.
Furthermore, a previous study~\cite{zollo2017debunking} reported that corrective information is primarily consumed by individuals who were already skeptical of the misinformation, and often fails to reach those who actually believe it, indicating that debunking also faces challenges in terms of reach.

On the basis of this background, a preventive intervention known as \textit{prebunking} has recently attracted attention in the field of psychology. 
Prebunking is an intervention rooted in (psychological) inoculation theory~\cite{mcguire1961relative}, and aims to strengthen individuals' cognitive resistance to misinformation by presenting weakened doses of misinformation (e.g., expected false claims with preemptive refutations) or by teaching common manipulation techniques (e.g., fake experts, logical fallacies, conspiracy theories, etc.) before they encounter actual misinformation~\cite{van2017inoculating,lewandowsky2021countering}.
Prebunking has been implemented in various forms, including text messages, infographics, video ads, and game-based interventions, and numerous experimental studies have demonstrated its effectiveness in reducing individuals' misinformation susceptibility~\cite{roozenbeek2019fake,roozenbeek2020breaking,maertens2021long,tay2022comparison,roozenbeek2022technique}.
In particular, it has been reported that active inoculation, such as game-based interventions that let players learn manipulation techniques through gameplay, is more effective than passive inoculation, which simply presents prebunking messages via text, infographics, or videos~\cite{basol2021towards,roozenbeek2020prebunking}.

However, while there are numerous studies focusing on the effectiveness of prebunking at the individual level and its effective design, an important open question remains: 
\textit{How can we identify the optimal targets for prebunking interventions to mitigate the spread of misinformation in a social network?}
In real-world social media platforms, it is often unrealistic to provide interventions to all users due to budget constraints.
Indeed, active inoculation through games requires motivating users to participate in the intervention and ensuring they understand the rules~\cite{kozyreva2024toolbox}.
While passive inoculation can be deployed at scale by broadcasting messages to all users, it still incurs costs. 
This is because achieving effective prebunking requires personalized interventions~\cite{barman2024personalised,barman2025rethinking} and follow-up efforts such as monitoring, iteration, and measuring outcomes~\cite{harjani2022practical}.
Therefore, to maximize its overall impact under limited intervention resources, strategic target selection based on network structure and diffusion dynamics is essential.

Based on this motivation, we formulate a combinatorial optimization problem, called the \textit{network prebunking problem}, which aims to select an optimal set of prebunking targets that minimize the spread of misinformation on a social network under limited intervention budgets.
To capture the realistic scenario in which users exposed to misinformation share either misinformation or corrective information depending on their susceptibility, we model the competitive diffusion of misinformation and corrective information based on the IC-N model~\cite{chen2011influence} and incorporate prebunking as an operation that reduces a user's susceptibility.
We prove that the network prebunking problem is NP-hard, and that its objective function is monotone and submodular.
This property implies that the greedy algorithm theoretically provides an approximation guarantee. 
However, since the greedy algorithm is computationally expensive and does not scale to large networks, we further propose an efficient approximation algorithm, MIA-NPP, based on the Maximum Influence Arborescence (MIA) approach~\cite{chen2010scalable_KDD}, which improves scalability by restricting influence propagation around each node to a local directed tree rooted at that node.
Through numerical experiments using real-world social network datasets, we demonstrate that MIA-NPP effectively suppresses the spread of misinformation compared to simple heuristics and approximation algorithms for related optimization problems, under both fully observed and uncertain model parameters.

Our key contributions are summarized as follows:
\begin{itemize}
    \item We are the first to formulate a combinatorial optimization problem for minimizing misinformation spread in a social network through prebunking, an approach distinct from conventional interventions such as \textit{blocking} and \textit{clarification} (detailed in Section~\ref{sect:related_work}).
    \item We prove that the network prebunking problem is NP-hard and that its objective function is monotone and submodular, providing a theoretical foundation for approximation guarantees of greedy algorithms.
    \item To address the scalability issue of the greedy approach, we propose an efficient approximation algorithm, MIA-NPP, based on the MIA framework.
    \item We conduct experiments on real-world social network datasets and show that MIA-NPP consistently outperforms baseline methods, even under parameter uncertainty.
\end{itemize}

\section{Related Work}
\subsection{Design of Effective Prebunking}

According to a practical guide on prebunking published by Google Jigsaw~\cite{harjani2022practical}, the effective design of prebunking interventions should incorporate three key elements: explicit presentation of manipulation techniques, concise and engaging formats, and appropriate timing and contextual relevance.

First, explicitly highlighting common manipulation techniques frequently used in misinformation, such as impersonation, emotional manipulation, polarization, and conspiratorial ideation, has been shown to be effective in boosting individuals' resistance to misinformation~\cite{roozenbeek2019the,roozenbeek2019fake}. 
Such technique-based inoculation does not rely on specific factual claims and can therefore be applied across diverse topics and contexts~\cite{traberg2022psychological,cook2017neutralizing,basol2021towards,roozenbeek2022technique}.
Second, short and visually engaging formats have been shown to be effective for prebunking interventions~\cite{roozenbeek2022psychological,basol2020good}. 
In particular, 30-second animated videos~\cite{roozenbeek2022psychological} or interactive and gamified interventions such as the Bad News game~\cite{roozenbeek2019the,roozenbeek2019fake} have been shown to facilitate better understanding and retention of manipulation techniques.
Third, the timing and contextual relevance of prebunking messages are also critical to their success~\cite{compton2021inoculation,roozenbeek2022psychological}. 
Prebunking is most effective when presented just prior to exposure to misinformation, or in contexts where such exposure is likely. 
In contrast, when there is a significant delay between intervention and exposure, its effectiveness may decay over time~\cite{capewell2024misinformation}. 
A field experiment on YouTube further demonstrates that inserting short prebunking videos directly before videos containing or related to misinformation significantly improves users' recognition of manipulation techniques and reduced their intent to share false content~\cite{roozenbeek2022psychological}.

Furthermore, in recent years, Bayiz and Topcu formulate the design of optimal prebunking strategies as a policy synthesis problem, aiming to optimize the timing and frequency of prebunking message delivery while minimizing disruption to the user experience~\cite{bayiz2023prebunking}.
However, while their work focuses on optimizing the frequency and timing of prebunking for a single user, with the goal of minimizing user experience disruption, our work aims to suppress the spread of misinformation across the entire social network by selecting optimal prebunking targets.

\subsection{Minimizaing Misinformation Spread on Social Networks}
\label{sect:related_work}

\begin{figure*}[t!]
    \centering
    \includegraphics[width=1\linewidth]{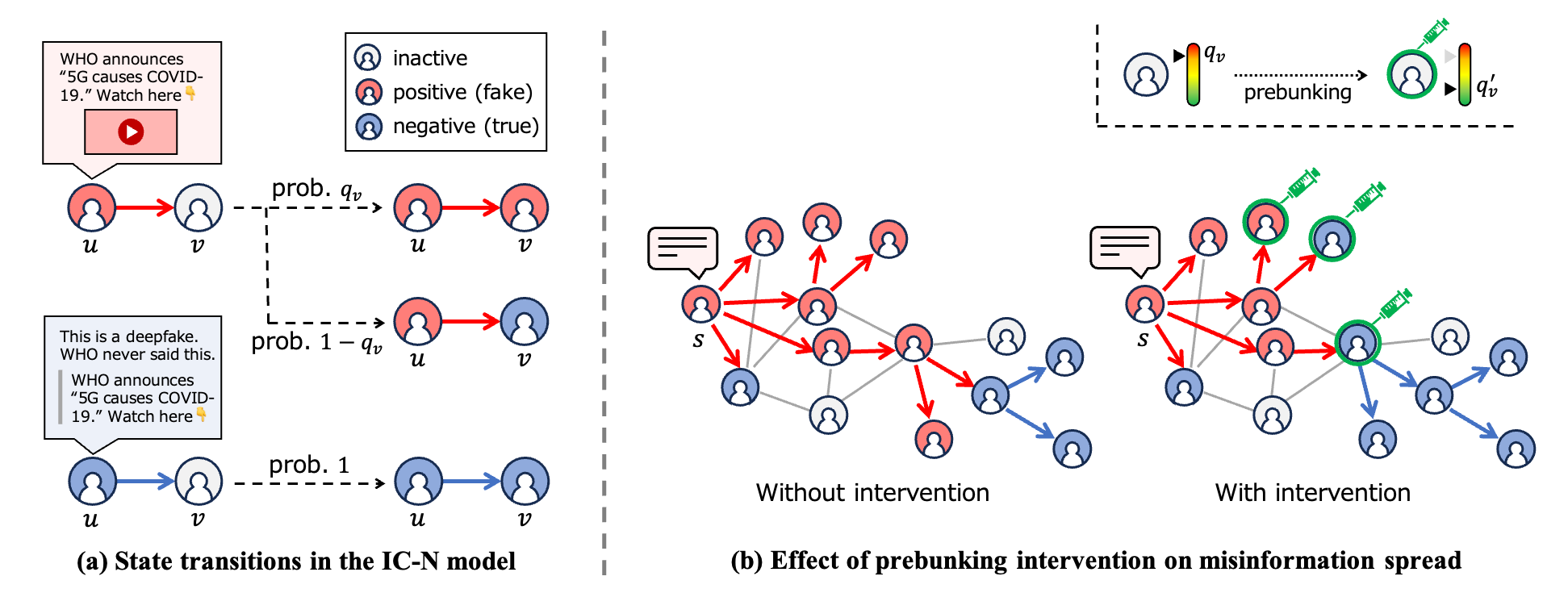}
    \caption{Schematic diagram of the network prebunking problem.}
    \label{fig:network_prebunking_problem}
\end{figure*}

In the existing literature, there are two approaches to suppressing the spread of misinformation on social networks: blocking and clarification~\cite{zareie2021minimizing,chen2022influence}.
The problem of suppressing the spread of misinformation using the blocking approach is often referred to as the Influence Minimization (IMIN) problem, while the problem of suppressing the spread of misinformation using the clarification approach is referred to as the Influence Blocking Maximization (IBM) problem.

The IMIN problem aims to suppress the spread of misinformation by blocking some nodes or edges in a network (i.e., excluding them from the network).
Formally, the IMIN problem by node blocking is a combinatorial optimization problem that finds a blocking node set $B \subset V \setminus S$ with size $|B|=k$ that minimizes the spread of misinformation when misinformation spreads from a given seed node set $S \subset V$ in a graph $G$ under some information diffusion model.
This problem is NP-hard under the Independent Cascade (IC) model~\cite{xie2023minimizing}, and several approximate solutions have been proposed, including greedy-based approaches~\cite{wang2013negative,fan2013least,yan2019minimizing,xie2023minimizing} and node centrality heuristics~\cite{yao2015topic}.
Additionally, the IMIN problem by edge blocking has also been studied and shown to be NP-hard~\cite{kimura2008solving,kimura2009blocking,yao2014minimizing,khalil2013cuttingedge}.
However, although the blocking approach is powerful, it can negatively affect the user experience, potentially leading to user complaints and deplatforming~\cite{wang2017drimux}.
In addition, ethical concerns such as freedom of expression and platform censorship have also been pointed out~\cite{hosni2019darim}.

The clarification approach is a more moderate method that attempts to neutralize and mitigate the influence of misinformation by appropriately injecting corrective information that counteracts misinformation into the network.
Formally, the IBM problem is a combinatorial optimization problem that aims to find a set of seed nodes $S_T$ with size $|S_T| = k$ for disseminating corrective information so as to maximize the obstruction of the spread of misinformation when misinformation spreads from a set of seed nodes $S_M$ in a graph $G$.
He et al.~\cite{he2012influence} first studied this problem under the competitive linear threshold model and proposed an efficient algorithm based on the MIA approach~\cite{chen2010scalable_KDD,chen2010scalable_ICDM}.
Budak et al.~\cite{budak2011limiting} introduced the Campaign-Oblivious Independent Cascade Model (COICM) as a competitive diffusion model of misinformation and corrective information, and proved that the IBM problem under the COICM is NP-hard.
They also demonstrated that a simple degree centrality heuristic achieves good empirical performance.
The IBM problem under the COICM has been frequently adopted in subsequent studies, and various approximation algorithms have been proposed, including  MIA approach-based~\cite{wu2017scalable}, centrality-based~\cite{arazkhani2019influence}, and community-based~\cite{lv2019community}.
Furthermore, recently, there have been several attempts to solve the IBM problem using machine learning-based methods~\cite{tong2020stratlearner,chen2023neural}.
However, as pointed out in the literature~\cite{budak2011limiting,wen2014shut}, the clarification approach is less effective if there is a significant delay between the onset of misinformation diffusion and the injection of corrective information.

This study proposes prebunking interventions as a novel approach for minimizing the spread of misinformation in social networks, which differs from blocking and clarification in its focus on reducing users' susceptibility prior to exposure, rather than removing nodes or injecting corrective information into the network.

\section{Problem Formulation}
\subsection{Information Diffusion Model}
\label{sect:diffusion_model}

In this study, we describe the competitive diffusion dynamics of misinformation and corrective information on a social network based on the IC-N model~\cite{chen2011influence}.
The IC-N model is a model extends the classic IC model by incorporating the probabilistic occurrence of negative opinions.

Let $G=(V, E)$ be a directed graph, where $V$ is a set of $n$ nodes, and $E \subseteq V \times V$ is a set of directed edges.
In the IC-N model, each node $v \in V$ takes one of three states: \textit{inactive}, \textit{positive (positively active)}, or \textit{negative (negatively active)}.
Each node $v$ is associated a parameter $q_v \in [0,1]$.
Note that, in the original study~\cite{chen2011influence}, $q_v$ was assumed to be the same for all nodes $v$ (i.e., $q_v=q$), whereas in our study, we allow $q_v$ to vary across nodes.
The state of node $v$ is determined by the state of its activating neighbor and its own parameter $q_v$.

For a given seed node set $S \subseteq V$, at step $t=0$, the nodes in $S$ are positively activated, and all other nodes are inactive.
At each discrete time step $t \ge 1$, if a node $u$ becomes positive (or negative) at that step, it has a single chance to activate its inactive neighbors $v \in \Nout_u$ at step $t+1$ with probability $p_{uv}^+$ (or $p_{uv}^-$), where $p_{uv}^+$ (resp. $p_{uv}^-$) denotes the probability that positive (resp. negative) information is propagated from $u$ to $v$ via the directed edge $(u, v) \in E$.
In this study, we assume that the propagation probability is independent of the type of information, i.e., $p_{uv}^+ = p_{uv}^- = p_{uv}$.
When an inactive node $v$ is activated by a positive neighbor, it becomes positive with probability $q_v$ and negative with probability $1-q_v$.
On the other hand, when it is activated by a negative neighbor, it becomes negative with probability $1$.
If an inactive node is activated simultaneously by positive and negative neighbors in the same step, the activation by the negative neighbor takes precedence.
Once a node $v$ becomes positive or negative, its state does not change thereafter.
The diffusion process terminates when no further activations are possible in the graph $G$.
The state transition process of the IC-N model is schematically illustrated in Fig.~\ref{fig:network_prebunking_problem}a.

In our problem setting, the inactive state corresponds to a state in which no information has been shared yet, the positive state corresponds to a state in which users have been deceived by misinformation and are sharing it, and the negative state corresponds to a state in which users have recognized the fallacy of the misinformation and are sharing correct information.
The parameter $q_v$ can be interpreted as the misinformation susceptibility of each user $v$.
That is, when a user $v$ receives misinformation (e.g., \textit{[Breaking] WHO Director-General announces ``5G causes COVID-19.'' Watch the video here: (Link to a deepfake video)}), the higher the susceptibility, the higher the likelihood of being deceived by the misinformation, and conversely, the lower the susceptibility, the higher the likelihood of recognizing the fallacy of the misinformation and adopting correct information.
On the other hand, if user $v$ receives corrective information (e.g., \textit{This is a deepfake. The WHO Director-General never said this.}) before receiving the misinformation, the user must believe the corrective information and will not adopt the original misinformation.
In reality, when an individual's ideological motives strongly influence their belief in information, they may reject corrective information even if they receive it and believe misinformation that aligns with their ideology.
However, such cases are beyond the scope of the model considered in this study.

\subsection{Network Prebunking Problem}
\label{sect:NPP}
In the IC-N model, we define prebunking for a node $v$ as an operation that reduces its misinformation susceptibility prior to the diffusion process, specifically updating $q_v$ to $q_v' = (1 - \varepsilon_v) q_v$, where $\varepsilon_v \in [0, 1]$ represents the individual intervention effect of prebunking on node $v$.
Let $X \subseteq V$ be the set of nodes selected for prebunking.
For given $S$ and $X$ in a graph $G$, let $ap^+(v;S,X,G)$ and $ap^-(v;S,X,G)$ be the probability that node $v$ eventually shares misinformation and corrective information, respectively.
We define the expected misinformation spread as $\sigma_G^+(S, X) = \sum_{v \in V} ap^+(v; S, X, G)$ which represents the expected number of nodes that share misinformation.
Similarly, the expected corrective information spread is defined as $\sigma_G^-(S, X) = \sum_{v \in V} ap^-(v; S, X, G)$.
We now consider a combinatorial optimization problem to find the optimal prebunking target set $X^\ast$ that minimizes the expected misinformation spread $\sigma_G^+(S, X)$, subject to an intervention cost constraint $|X| = k ~(\le n)$.
That is,
\begin{align}
    X^\ast = \argmin_{X \subseteq V} ~\sigma_G^+(S, X), \quad \text{s.t.}\quad |X| = k.
    \label{eq:NPP}
\end{align}
We call this problem the network prebunking problem.
Note that, the total expected spread $\sigma_G(S) = \sigma_G^+(S, X) + \sigma_G^-(S, X)$ is independent of the prebunking set $X$ for a given $S$.
Therefore, this problem can also be formulated as follows:
\begin{align}
    X^\ast = \argmax_{X \subseteq V} ~\sigma_G^-(S, X), \quad \text{s.t.}\quad |X| = k.
    \label{eq:NPP2}
\end{align}
Although we here assume a fixed seed set $S$ in Eqs.~(\ref{eq:NPP})-(\ref{eq:NPP2}), the formulation can be naturally extended to the case where $S$ is drawn from a probability distribution. 
The details are given in Appendix~\ref{sec:NPP-US}.

Figure~\ref{fig:network_prebunking_problem} shows a schematic diagram of the network prebunking problem.
Without prebunking intervention, misinformation disseminated from seed node $s$ spreads across the network via edges, causing many nodes to share the misinformation.
On the other hand, with prebunking intervention on some nodes, the intervened nodes are less likely to be deceived by the misinformation and more likely to share corrective information.
As a result, the spread of misinformation across the network can be reduced compared to the case without intervention.
The network prebunking problem aims to optimally select intervention targets and minimize the final spread of misinformation when the number of nodes that can be intervened (intervention budget) is given as $k$.

The network prebunking problem is closely related to both the IMIN and IBM problems.
Specifically, when $q_v = \varepsilon_v = 1$ for all nodes $v \in V$, and the propagation probabilities for each edge $(u,v) \in E$ are set as $p_{uv}^+ = p_{uv}$ and $p_{uv}^- = 0$, the network prebunking problem becomes equivalent to the IMIN problem by node blocking.
Similarly, under the same assumption $q_v = \varepsilon_v = 1$, the problem corresponds to the IBM problem with node-specific delay parameters.
For details, see Appendix~\ref{sect:other_problems}.

\subsection{Hardness of Network Prebunking Problem}
Here, we prove that the network prebunking problem is NP-hard.
To do so, we show that the set cover problem, which is known to be NP-complete, can be reduced to a special case of the network prebunking problem.
The set cover problem is defined as follows:
Given a universe set $U=\{u_1, \dots, u_n\}$, and a family of subsets $\mathcal{T} =\{T_1, \dots, T_m\}$ such that $U=\bigcup_{T_j \in \mathcal{T}}T_j$, determine whether $U$ can be covered by $k\,(<m)$ subsets.

\begin{theorem}
    The network prebunking problem is NP-hard.
\end{theorem}

\begin{proof}
For an instance $(U, \mathcal{T}) = (\{u_i\}_{i=1}^n, \{T_j\}_{j=1}^m)$ of the set cover problem, construct a directed graph $H = (V_H, E_H)$ consisting of $n + m + 1$ nodes using the following procedure.
The node set $V_H$ consists of nodes $c_i$ corresponding to each element $u_i$ of the universe set $U$, nodes $b_j$ corresponding to each element set $T_j$ of the family of subsets $\mathcal{T}$, and a dummy node $a$.
For each $j =1, 2, \dots, m$, we connect dummy node $a$ to $b_j$ with a directed edge $(a, b_j)$, and connect node $b_j$ to node $c_i$ with a directed edge $(b_j, c_i)$ if $u_i \in T_j$.
Furthermore, for any edge $e \in E_H$ in the graph $H$, set $p_e = 1$, and for any node $v \in V_H$, set $q_v = 1$ and $\varepsilon_v = 1$.
Figure~\ref{fig:SetCover_instance} shows the graph $H$ corresponding to a set cover instance $(U, \mathcal{T})$ with $U=\{u_1, \dots, u_5\}$ and $\mathcal{T}=\{\{u_1, u_2, u_3\}, \{u_2, u_4\}, \{u_3, u_4\}, \{u_4, u_5\}\}$.

Consider the network prebunking problem defined in Eq.~(\ref{eq:NPP2}) on graph $H$ with seed set $S = \{a\}$.
When $X = \emptyset$, all nodes deterministically become positive after the diffusion process.
In contrast, if we intervene on a node $v \in V_H \setminus S$, then $v$ and all of its downstream nodes $w \in \Nout_v$ deterministically become negative.
Note that, since intervening on a leaf node $c_i$ only affects that node itself; hence to maximize the corrective information spread $\sigma_H^-(S, X)$, optimal intervention targets must be selected from nodes $\{b_1, b_2, \dots, b_m\}$.
From the above, solving the set cover problem is equivalent to determining whether there exists a set $X$ such that $\sigma_H^-(S, X) \ge k + n$.
Therefore, the network prebunking problem is NP-hard.
\end{proof}

\begin{figure}[t]
    \centering
    \includegraphics[width=0.6\linewidth]{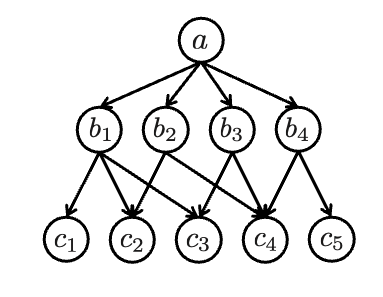}
    \caption{Graph $H$ corresponding to a set cover instance with $U=\{u_1, \dots, u_5\}$ and $\mathcal{T}=\{\{u_1, u_2, u_3\}, \{u_2, u_4\}, \{u_3, u_4\}, \{u_4, u_5\}\}$.}
    \label{fig:SetCover_instance}
\end{figure}

\subsection{Submodularity of Corrective Information Spread}

To develop an approximation algorithm for the network prebunking problem, 
we first need to show that its objective function $\sigma_G^-(S, X)$ is submodular with respect to $X$. 
This property provides a theoretical basis for the design of approximation heuristics discussed in the next section.

In what follows, we formally prove the submodularity of $\sigma_G^-(S, X)$, starting with the following lemma.

\begin{lemma}
    \label{lem:lemE1}
    For any given submodular function $f(X)$ and any fixed set $W \subseteq V$, $f(X \cap W)$ is also submodular.
\end{lemma}
\begin{proof}
Let $g(X) = f(X \cap W)$, and define $\Delta_g(x |A) := g(A \cup \{x\}) - g(A)$.
For any pair of sets $A \subseteq B \subseteq V$ and any element $x \in V$, we show that $\Delta_g(x | A) \ge \Delta_g(x | B)$ holds.
By definition, we have
\begin{align*}
    \Delta_g(x|A) &= f((A \cup \{x\}) \cap W) - f(A \cap W) \\
    &= f((A \cap W) \cup (\{x\} \cap W)) - f(A \cap W).
\end{align*}
For $x \notin W$, since $\Delta_g(x|A) = \Delta_g(x|B) = 0$, it follows trivially that $\Delta_g(x | A) \ge \Delta_g(x|B)$.
For $x \in W$, we have $\Delta_g(x|A) = f((A \cap W) \cup \{x\}) - f(A \cap W) = \Delta_f(x|A \cap W)$.
Similarly, $\Delta_g(x|B) = \Delta_f(x|B \cap W)$.
Since $A \cap W \subseteq B \cap W$ and $f(\cdot)$ is submodular, it follows that $\Delta_g(x | A) = \Delta_f(x|A \cap W) \ge \Delta_f(x|B \cap W) = \Delta_g(x|B)$.
Therefore, $g(X) = f(X \cap W)$ is submodular.
\end{proof}

\begin{theorem}
    The corrective information spread $\sigma_G^-(S, X)$ is monotone and submodular with respect to $X$.
\end{theorem}
\begin{proof}
Following the live-edge model~\cite{kempe2003maximizing}, for each edge $(u,v) \in E$ in graph $G$, we independently determine whether the edge is live with probability $p_{uv}$ and construct the live-edge graph $L = (V, E_L)$.
Furthermore, for each node $v \in V$, we independently sample and fix a uniform random number $r_v \sim \mathrm{Unif}[0,1]$.
For a given realization $\omega = (L, \{r_v\}_{v \in V})$, the propagation process under the IC-N model on the live-edge graph $L$ becomes deterministic.
Then, the influence from the seed node set $S$ propagates to all nodes reachable from $S$ on the live-edge graph $L$.
For a given $X$, when each node $v$ first receives influence from a positive neighbor, it becomes positive if $r_v \le q_v^X$ and negative if $r_v > q_v^X$, where $q_v^X$ is the post-intervened susceptibility defined as follows:
\begin{align}
    q_v^X = 
    \begin{cases}
        (1 - \varepsilon_v) q_v, & \text{if $v \in X$,} \\
        q_v, & \text{otherwise.}
    \end{cases}
\end{align}

On the live-edge graph $L$, let $d_L(S, v)$ denote the shortest distance from $S$ to $v$.
If $v$ is unreachable from $S$, set $d_L(S, v) = \infty$.
First, let $R_L(S) = \{u \in V \mid d_L(S,u) < \infty\}$ denote the set of nodes reachable from $S$ on $L$.
For a given $\omega$ and $X$, let $F_\omega(X) = \{u \in R_L(S) \mid r_u > q_u^X\}$ denote the set of nodes in $L$ that are necessarily negative due to the influence of $S$.
Furthermore, let $\mathcal{C}_L(u) = \{v \in V \mid d_L(S, v) = d_L(S, u) + d_L(u, v)\}$ denote the set of nodes $v \in V$ such that the shortest path from $S$ to $v$ on $L$ includes $u$.
According to the rules of the IC-N model, if a node $u$ becomes negative, then every node $v \in \mathcal{C}_L(u)$ also becomes negative. 
Thus, letting $N_\omega(X)$ denote the set of nodes that eventually become negative, we have $N_\omega(X) = \bigcup_{u \in F_\omega(X)} \mathcal{C}_L(u)$.
Define $A_\omega = \{u \in R_L(S) \mid r_u > q_u\}$ and $B_\omega = \{u \in R_L(S) \mid (1 - \varepsilon_u) q_u < r_u \le q_u\}$, then $F_\omega(X) = A_\omega \cup (X \cap B_\omega)$.
Thus, 
\begin{align*}
    |N_\omega(X)| &= \left|\bigcup_{u \in F_\omega(X)} \mathcal{C}_L(u)\right|
    = \left|\bigcup_{u \in A_\omega} \mathcal{C}_L(u) \cup \bigcup_{u \in X \cap B_\omega} \mathcal{C}_L(u)\right| \\
    &= \left|\bigcup_{u \in A_\omega} \mathcal{C}_L(u)\right| + \left|\bigcup_{u \in X \cap B_\omega} H_\omega(u)\right|,
\end{align*}
where we define $H_\omega(u) = \mathcal{C}_L(u) \setminus \bigcup_{u \in A_\omega} \mathcal{C}_L(u)$.
In the above equation, the first term on the right-hand side does not depend on $X$. 
The second term, $h(X \cap B_\omega) = |\bigcup_{u \in X \cap B_\omega}H_\omega(u)|$, is the restriction of the covering function $h(X)$ to the set $B_\omega$.
Thus, by Lemma~\ref{lem:lemE1}, $|N_\omega(X)|$ is submodular.
Therefore, the corrective information spread $\sigma_G^-(S, X) = \mathbb{E}_\omega[|N_\omega(X)|]$ is also submodular.
Monotonicity is trivial.
\end{proof}

This theorem indicates that, although the network prebunking problem is NP-hard, it can be theoretically approximated by the greedy algorithm with an approximation ratio of $(1 - 1/e - \varepsilon_{\mathrm{MC}})$, where $\varepsilon_{\mathrm{MC}}$ represents the estimation error introduced by Monte Carlo (MC) simulations used to evaluate $\sigma_G^-(S, X)$~\cite{nemhauser1978analysis, kempe2003maximizing}.

\section{Approximation Algorithm}
\label{sect:approx_alg}

As mentioned in the previous section, while the greedy algorithm provides a theoretical approximation guarantee for the network prebunking problem~\cite{nemhauser1978analysis, kempe2003maximizing}, it is computationally expensive due to repeated MC evaluations of $\sigma_G^-(S, X)$. 
Moreover, achieving a sufficiently small $\varepsilon_{\mathrm{MC}}$ requires an extremely large number of MC samples. 
This issue is particularly pronounced in our setting, because an intervention on a node $v$ reduces its susceptibility by a fraction of $\varepsilon_v$, and thus the marginal gain obtained by each intervention is generally very small. 
As a result, when the number of MC samples is insufficient, the simulation noise can easily obscure such small differences in gain, leading to inaccurate greedy selections.

To address this challenge, inspired by the design of CMIA-O~\cite{wu2017scalable} for the IBM problem, we propose an approximation algorithm, MIA-NPP, based on the MIA model~\cite{chen2010scalable_KDD}.
The MIA model restricts influence propagation around each node to a local directed tree rooted at that node, called an \textit{arborescence}, thereby significantly improving the efficiency of influence spread estimation.

In the MIA model, for each node $v$, two types of arborescences are considered: the Maximum Influence In-Arborescence (MIIA) and the Maximum Influence Out-Arborescence (MIOA).
To construct MIIA and MIOA, we first define the Maximum Influence Path (MIP) in the graph $G$.
For a path from node $u$ to node $v$, $P_{uv} = ((u, u_1), (u_1, u_2), \dots, (u_l, v))$, the probability that the influence of node $u$ propagates to node $v$ through a path $P_{uv}$ is given by $pp(P_{u v}) = \prod_{e \in P_{uv}} p_{e}$.
Then, the MIP from $u$ to $v$ is defined as:
\begin{align}
\MIP_{uv} = \argmax_{ P_{u v} \in \mathcal{P}_{uv}} ~pp(P_{u v}),
\label{eq:MIP}
\end{align}
where $\mathcal{P}_{uv}$ denotes the set of all paths from $u$ to $v$.
If no path exists from $u$ to $v$, then we define $\MIP_{uv} = \emptyset$.
The MIP can be efficiently computed by assigning edge weights $w_{uv} = -\log p_{uv}$ to each edge $(u, v) \in E$ and applying Dijkstra's algorithm to find the shortest path from $u$ to $v$ in this weighted graph.

Given an influence threshold $\theta$, the MIIA of a node $v$, denoted by $\MIIA(v)$, is defined as the subgraph of $G$ induced by the union of MIPs ending at $v$ whose propagation probabilities satisfy $pp(\MIP_{uv}) \ge \theta$.
Similarly, the MIOA of node $v$, denoted by $\MIOA(v)$, is defined as the subgraph of $G$ induced by the union of MIPs starting from $v$ that satisfy $pp(\MIP_{vu}) \ge \theta$.
Intuitively, $\MIIA(v)$ estimates the local influence toward node $v$ from other nodes in the graph, while $\MIOA(v)$ estimates the local influence from node $v$ to others.
The influence threshold $\theta$ controls the size of the local influence region around a node $v$.

Following the spirit of the MIA model, we assume that the influence to a node $u \in V \setminus S$ propagates only along the edges in $\MIIA(u)$.
For given $S$ and $X$, let $ap^+(v; S, X, \MIIA(u))$ denote the probability that a node $v$ in $\MIIA(u)$ is in the positive state at the end of the diffusion process.
We also denote by $ap^+_t(v; S, X, \MIIA(u))$ the probability that node $v$ is in the positive state at time step $t$, and by $\pi_t^+(v; S, X, \MIIA(u))$ the probability that node $v$ becomes positive for the first time at time $t$.
Note that, for brevity, we omit full arguments and simply write $ap^+(v)$, $ap_t^+(v)$, and $\pi_t^+(v)$ when the context is clear.
Similarly, we define the probabilities $ap^-(v)$, $ap^-_t(v)$, and $\pi_t^-(v)$ for the negative state.
The probability that node $v$ becomes either positive or negative for the first time at time $t$ is denoted by $\pi_t(v)$, which satisfies $\pi_t(v) = \pi_t^+(v) + \pi_t^-(v)$.
Moreover, assuming that influence propagation within $\MIIA(u)$ terminates at step $T$, we have $ap^+(u) = ap^+_T(u)$ and $ap^-(u) = ap^-_T(u)$.

In $\MIIA(u)$, a node $v$ becomes positive for the first time at step $t+1$ if (i) it is inactive (i.e., neither positive nor negative) at time $t$ and (ii) is positively activated by a neighbor that becomes positive at time $t$.
According to the definition of the IC-N model, for a node $v$ to become positive by its neighbors, it must receive only misinformation and no corrective information from its neighbors, and be deceived by it with probability $q_v^X$.
Therefore, $\pi_{t+1}^+(v)$ is written as
\begin{align}
    \pi_{t+1}^+(v) &= q_v^X \left\{\left(1 - \beta_{t}(v)\right) - \left(1 - \beta_{t}^-(v)\right) \right\} \notag \\
    & \quad \times(1 - ap^{+}_{t}(v)) (1 - ap^{-}_{t}(v)),
    \label{eq:pi_p}
\end{align}
where $\beta_{t}(v) = \prod_{w \in \Nin_v} \left(1 - \pi_{t}(w) p_{wv}\right)$ denotes the probability that node $v$ does not receive any information from its neighbors at time $t+1$, and $\beta_{t}^{-}(v) = \prod_{w \in \Nin_v} \left(1 - \pi_{t}^{-}(w) p_{wv}\right)$ denotes the probability that node $v$ does not receive corrective information from its neighbors at time $t+1$.
On the other hand, $\pi_{t+1}^-(v)$ is written as
\begin{align}
    \pi_{t+1}^-(v) &= \left\{q_v^X \left(1 - \beta_{t}^-(w) \right) + (1-q_v^X)  \left(1 - \beta_{t}(w) \right) \right\} \notag \\
    & \quad \times(1 - ap^{+}_{t}(v)) (1 - ap^{-}_{t}(v)),
    \label{eq:pi_n}
\end{align}
since $\pi_{t+1}^-(v) =  \pi_{t+1}(v) - \pi_{t+1}^+(v)$.
Moreover, it follows from the definition that
\begin{align}
    &ap^+_{t+1}(v) = ap^+_{t}(v) + \pi_{t+1}^+(v), \label{eq:pap_t}\\
    &ap^-_{t+1}(v) = ap^-_{t}(v) + \pi_{t+1}^-(v). \label{eq:nap_t}
\end{align}
The initial boundary conditions are defined as follows:
for any node $v \in S$, we set $\pi_0^+(v) = ap^+_0(v) = 1$ and $\pi_0^-(v) = ap^-_0(v) = 0$;
for any node $v \notin S$, we set $\pi_0^+(v) = \pi_0^-(v) = ap^+_0(v) = ap^-_0(v) = 0$.
Based on the recurrence equations (\ref{eq:pi_p})–(\ref{eq:nap_t}) and the boundary conditions described above,
the probability $ap^+(v; S, X, \MIIA(u))$ can be computed recursively via dynamic programming as outlined in Algorithm~\ref{alg:pap}.
For notational simplicity, the pseudocode refers to the node set of $\MIIA(u)$ simply as $\MIIA(u)$, omitting the explicit notation $V(\MIIA(u))$ (the same convention applies to $\MIOA(u)$).

In Algorithm~\ref{alg:pap}, $Z_t$ denotes the set of nodes that may potentially receive information at step $t$ from nodes in $Z_{t-1}$.
The function $\mathrm{child}(w)$ returns the child of node $w$ in $\MIIA(u)$, if one exists; 
by definition, each node has at most one child in $\MIIA(u)$.

The auxiliary variable $\TMP(v)$ stores the probability that node $v$ remains inactive (i.e., neither positively nor negatively activated) at step $t$,
while $\TMPN(v)$ stores the probability that node $v$ is not negatively activated at step $t$.
These values are used to compute the activation probabilities $\pi_{t+1}^+(v)$ and $\pi_{t+1}^-(v)$.
In the main loop, computing $ap_t^+(u)$ requires at most one traversal of the path from seed nodes to $u$ in $\MIIA(u)$,
and since each edge activation attempt takes constant time, the overall time complexity of Algorithm~\ref{alg:pap} is $O(n_\theta)$,
where $n_\theta$ is the maximum size of $\MIIA(u)$ or $\MIOA(u)$ across all $u \in V$.

By using the positive activation probability $ap^+(v; S, X, \MIIA(v))$, we can design an approximation algorithm, MIA-NPP, for selecting prebunking targets
The influence of adding node $w$ to $X$ on the value $ap^+(v; X) = ap^+(v; S, X, \MIIA(v))$ for node $v$ from node $w$ is calculated as
\begin{align}
\Delta(w, v, X) = \left| ap^+(v; X) - ap^+(v; X \cup \{w\}) \right|.
\end{align}
Thus, by summing this quantity over all $v$, we can estimate the influence of node $w$ on the positive activation probabilities of other nodes.
MIA-NPP selects prebunking nodes by iteratively adding the node $w$ that maximizes $\Delta(w, X) = \sum_v \Delta(w, v, X)$ to the current set $X$.
A key advantage of the MIA model is that, when a node $u$ is newly added to $X$, it suffices to update $\Delta(w, X)$ only for nodes $v \in \MIOA(u)$ potentially influenced by $u$, and for nodes $w \in \MIIA(v)$ that may influence those $v$.
This update is carried out by subtracting the previous value of $\Delta(w, v, X)$ (before adding $u$) from $\Delta(w, X)$ and then adding the new value after $u$ is included.
This localized update strategy significantly reduces computation time compared to recomputing $ap^+(v; S, X, \MIIA(v))$ for all $v$ in each iteration.
The complete pseudocode for MIA-NPP is presented in Algorithm~\ref{alg:MIA-NPP}.

\begin{algorithm}[t]
  \caption{$ap^+(u; S, X, \MIIA(u))$}
  \label{alg:pap}
  \SetKwInOut{Input}{Input}
  \SetKwInOut{Output}{Output}
  \DontPrintSemicolon

  set $Z_0 = S \cap \MIIA(u)$, $Z_t = \emptyset$ for $t \ge 2$\;
  initialize $\pi_0^+(v), \pi_0^-(v), ap_0^+(v), ap_0^-(v)$ for all $v \in \MIIA(u)$ according to boundary conditions\;
  set $t = 0$\;
  \While{$Z_t \neq \emptyset$}{
    set $\TMP(v) = \TMPN(v) = 1$ for all $v \in \MIIA(u) \setminus S$\;
    \For{$w \in Z_t$}{
      $v = \mathrm{child}(w)$\;
      $Z_{t+1} \leftarrow Z_{t+1} \cup \{v\}$\;
      $\TMP(v) \leftarrow \TMP(v) \times (1 - (\pi_t^+(w) + \pi_t^-(w))p_{wv})$\;
      $\TMPN(v) \leftarrow \TMPN(v) \times (1 - \pi_t^-(w)p_{wv})$\;
    }
    \For{$v \in Z_{t+1}$}{
      $\pi_{t+1}^+(v) = q_v^X \left\{(1 - \TMP(v)) - (1 - \TMPN(v))\right\}(1 - ap_t^+(v)) (1 - ap_t^-(v))$\;
      $\pi_{t+1}^-(v) = \left\{ q_v^X (1 - \TMPN(v)) + (1 - q_v^X) (1 - \TMP(v))\right\} (1 - ap_t^+(v)) (1 - ap_t^-(v))$\;
      $ap_{t+1}^+(v) = ap_t^+(v) + \pi_{t+1}^+(v)$\;
      $ap_{t+1}^-(v) = ap_t^-(v) + \pi_{t+1}^-(v)$\;
    }
    $t \leftarrow t + 1$\;
  }
  \Return{$ap_t^+(u)$}
\end{algorithm}

\begin{algorithm}[t]
  \caption{MIA-NPP$(G, S, k, \theta)$}
  \label{alg:MIA-NPP}
  \DontPrintSemicolon

  set $X = \emptyset$, $U = \emptyset$\;
  set $\Delta(v) = 0$ for $v \in V$\;
  \For{$s \in S$}{
    construct $\MIOA(s)$\;
    $U \leftarrow U \cup \MIOA(s)$\;
  }
  \For{$u \in U$}{
    construct $\MIIA(u)$\;
    compute $ap^+(u; S, X, \MIIA(u))$\;
    \For{$v \in \MIIA(u)$}{
      compute $ap^+(u; S, X \cup \{v\}, \MIIA(u))$\;
      $\Delta(v, u, X) = |ap^+(u; S, X, \MIIA(u)) - ap^+(u; S, X \cup \{v\}, \MIIA(u))|$\;
      $\Delta(v) \leftarrow \Delta(v) + \Delta(v, u, X)$\;
    }
  }
  \For{$i = 1, 2, \dots, k$}{
    $u = \arg\max_{v \in U} \Delta(v)$\;
    construct $\MIOA(u)$\;
    \For{$v \in \MIOA(u) \cap U$}{
      \For{$w \in \MIIA(v)$}{
        compute $ap^+(v; S, X \cup \{w\}, \MIIA(v))$\;
        $\Delta(w, v, X) = |ap^+(v; S, X, \MIIA(v)) - ap^+(v; S, X \cup \{w\}, \MIIA(v))|$\;
        $\Delta(w) \leftarrow \Delta(w) - \Delta(w, v, X)$\;
      }
    }
    $X \leftarrow X \cup \{u\}$\;
    $U \leftarrow U \setminus \{u\}$\;
    \For{$v \in \MIOA(u) \cap U$}{
      compute $ap^+(v; S, X, \MIIA(v))$\;
      \For{$w \in \MIIA(v)$}{
        compute $ap^+(v; S, X \cup \{w\}, \MIIA(v))$\;
        $\Delta(w, v, X) = |ap^+(v; S, X, \MIIA(v)) - ap^+(v; S, X \cup \{w\}, \MIIA(v))|$\;
        $\Delta(w) \leftarrow \Delta(w) + \Delta(w, v, X)$\;
      }
    }
  }
  \Return{$X$}
\end{algorithm}

\paragraph{Complexity analysis of MIA-NPP}
Let $n = |V|$ and $n_s = |S|$.
We define $n_{\theta}$ and $t_{\theta}$ as the maximum size and the maximum construction time of any $\MIOA(u)$ or $\MIIA(u)$, respectively.
Since the construction time is linear in the size, it holds that $n_{\theta} = O(t_{\theta})$.
The candidate set $U$ for prebunking targets has size $O(n_s n_{\theta})$.
As previously discussed, the time complexity of computing $ap^+(u, S, X, \MIIA(u))$ using Algorithm~\ref{alg:pap} is $O(n_{\theta})$.
A max-heap is employed to store $\Delta(v)$ values, which can be initialized in $O(n)$ time.
Therefore, the time complexity of lines 6--12 is $O(n_s n_{\theta}(t_{\theta} + n_{\theta}(n_{\theta} + n)))$, which simplifies to $O(n_s t_{\theta}^2(t_{\theta} + n))$.
In the main loop, selecting a new prebunking node from the max-heap requires $O(1)$ time, and updating each $\Delta(w)$ incurs $O(\log n)$ time.
Thus, the total complexity of the main loop is $O(k (t_{\theta} + n_{\theta}^2(n_{\theta} + \log n)))$, simplifying to $O(kt_{\theta}^2(t_{\theta} + \log n))$.
Accordingly, thel time complexity of MIA-NPP is $O(n_s t_\theta^2 (t_\theta + n) + kt_\theta^2(t_\theta + \log n))$.

For space complexity, the algorithm maintains the candidate set $U$,
the values $ap^+(u, S, X, \MIIA(u))$ for each $u \in U$ and $v \in \MIIA(u)$,
the $\MIIA(u)$ subgraph for each $u \in U$,
and $\Delta(v)$ for each $v \in V \setminus S$.
Thus, the overall space complexity is $O(n_s n_{\theta}^2 + n)$.

\paragraph{Comparison with greedy baseline.}
To position MIA-NPP against the theoretical greedy benchmark, we evaluated it on small synthetic graphs, where the CELF implementation~\cite{leskovec2007cost} of the greedy algorithm can be executed within a reasonable time budget. 
The results show that MIA-NPP achieves comparable performance to the greedy method. 
In terms of efficiency, MIA-NPP completes within roughly 20 seconds on average, whereas the greedy algorithm requires about two orders of magnitude more time. 
Although MIA-NPP does not provide a strict approximation guarantee, it empirically achieves near-greedy performance with substantially lower computational cost.
Detailed results and runtime statistics are provided in Appendix~\ref{appendix:compare_greedy}.

\section{Numerical Experiments}
In this section, we conduct numerical experiments on real-world network datasets to address the following research questions:
\textbf{RQ1}: How effective is MIA-NPP in suppressing the spread of misinformation compared to simple heuristics and approximation algorithms for related problems?
\textbf{RQ2}: Does MIA-NPP remain effective when the misinformation susceptibility $q_v$ and the individual intervention effect $\varepsilon_v$ of each node are uncertain?

We implemented all algorithms in Python 3.10 and ran the experiments on a MacBook Pro (Apple M1 Max, 64GB RAM, macOS Sonoma 14.4.1) without using GPU acceleration or parallelization.
All data and code used in our experiments are publicly available\footnote{\url{https://github.com/s-furutani/network-prebunking-problem}}.

\subsection{Datasets}
In this experiment, we utilize the UPFD dataset~\cite{dou2021user} to construct a social network in which each node has the susceptibility parameter.
UPFD is a derived dataset built on top of FakeNewsNet~\cite{shu2018fakenewsnet}, containing news diffusion networks along with veracity labels (true or false) for news articles, annotated by PolitiFact or GossipCop. 
Each diffusion network consists of a root node representing a news article and user nodes representing Twitter users who shared it.
Edges in the network capture the propagation paths of the news: 
edges between the root node and user nodes represent direct tweets containing the news article, while edges between user nodes indicate retweets.
There are 314 diffusion networks from PolitiFact (157 labeled as fake) and 5,464 from GossipCop (2,732 labeled as fake). 
To construct a unified diffusion network, we merge the edge sets of all individual networks and treat all root nodes (i.e., news articles) as a single node. 
We refer to the resulting networks as the PolitiFact network and the GossipCop network, respectively. 
The PolitiFact network contains 30,813 nodes and 33,488 edges, while the GossipCop network consists of 75,915 nodes and 85,308 edges.

\subsection{Baselines}
Since no existing algorithm is to solve the network prebunking problem, we compare our proposed algorithm, MIA-NPP, with the following simple heuristics and approximation algorithms for related problems:
\begin{itemize}
    \item \textbf{Random}: selects $k$ nodes uniformly at random.
    \item \textbf{Gullible}: selects the $k$ nodes with the highest misinformation susceptibility $q_v$.
    \item \textbf{Degree}: selects the $k$ nodes with the highest out-degree.
    \item \textbf{Distance}: selects the $k$ nodes closest to the seed nodes in terms of shortest path distance, where edge weights are defined as $w_{uv} = -\log p_{uv}$.
    \item \textbf{AdvancedGreedy}~\cite{xie2023minimizing}: is an approximation algorithm for the influence minimization problem under the IC model, based on graph sampling and dominator trees. The nodes selected for blocking nodes by AdvancedGreedy are used as prebunking targets.
    \item \textbf{CMIA-O}~\cite{wu2017scalable}: is an approximation algorithm for the influence blocking maximization problem under the COICM model. The seed nodes selected for spreading corrective information by CMIA-O are used as prebunking targets.
\end{itemize}

\subsection{Parameter Settings}
\label{sect:param_setting}

In this experiment, we set the propagation probability $p_{uv}$ of each edge $(u, v)$ to be proportional to the number of news items shared by user $v$, denoted $n_v^{\mathrm{share}}$, as $p_{uv} = c \cdot n_v^{\mathrm{share}} / D$, where $D$ is the total number of news items, and the proportionality constant $c$ is set to $c=30$ based on preliminary simulations, such that approximately $10\%$ nodes of the network is activated on average when the information propagates from the root node under the IC-N model.
The seed set $S$ consists only of the root node $r$, i.e., $S = \{ r \}$.
The susceptibility $q_v$ for each user node $v$ is computed as $q_v = (n_v^{\mathrm{fake}} + 1)/(n_v^{\mathrm{share}} + 2)$, where $n_v^{\mathrm{fake}}$ is the number of fake news items shared by the user $v$. 
Note that since most users share only one news item, we apply additive smoothing to avoid extreme values of $q_v$. 
Without smoothing (i.e., using the raw ratio $q_v = n_v^{\mathrm{fake}} / n_v^{\mathrm{share}}$), the resulting values of $q_v$ will be either $0$ or $1$ for many users.
The individual intervention effect $\varepsilon_v$ for each user node $v$ is sampled from a truncated normal distribution $\mathcal{N}_{[0, 1]}(\mu_\varepsilon, \sigma_\varepsilon^2)$, with mean $\mu_{\varepsilon} = 0.5$ and variance $\sigma_{\varepsilon}^2 = 0.1$.
The susceptibility and individual intervention effect of root node $r$ are respectively set to $q_r=1$ and $\varepsilon_r=0$.

For AdvancedGreedy, we set the number of sampled graphs to $\rho = 100$. 
For both CMIA-O and MIA-NPP, the influence threshold parameter for constructing MIIA and MIOA structures is set to $\theta = 0.001$.
For MIA-NPP, we confirmed that the choice of $\theta$ has minimal impact on performance, while larger values of $\theta$ reduce execution time (see Appendix~\ref{sect:sensitivity_analysis}).

\subsection{Results}

MIA-NPP utilizes information on each node's susceptibility $q_v$ and individual intervention effect $\varepsilon_v$ to select prebunking targets.
In this section, we evaluate the effectiveness of MIA-NPP under two conditions: when $q_v$ and $\varepsilon_v$ are fully observable (Section~\ref{sect:eval_perfect}), and when they are uncertain (Section~\ref{sect:eval_uncertain}).

To assess the effectiveness of each algorithm in suppressing the spread of misinformation, we conducted $1000$ diffusion simulations per algorithm, applying interventions to $k \in \{0, 20, 40, \dots, 200\}$ nodes selected by each algorithm.
We computed the relative spread $y(k)/y(0)$, where $y(k)$ is the number of nodes who eventually shared misinformation averaged over 1000 runs with an intervention of $k$ nodes.
This metric indicates the remaining extent of misinformation relative to the no-intervention case ($k = 0$); 
a smaller value implies more effective suppression of misinformation.

\subsubsection{Evaluation under perfect observations (RQ1)} 
\label{sect:eval_perfect}

First, we evaluate the effectiveness of MIA-NPP in the case where $q_v$ and $\varepsilon_v$ for each node are fully available for target selection.
Figure~\ref{fig:results_real} shows the misinformation suppression effect of each algorithm on the PolitiFact and GossipCop networks.
As shown, the Random baseline fails to suppress misinformation, indicating that randomly selecting prebunking targets is insufficient.
The Gullible, which utilize the node's susceptibility, and the Degree and the Distance, which refer to the network structure, can suppress misinformation compared to the Random, but the suppression effect is relatively small.
The approximation algorithms for related problems, AdvancedGreedy and CMIA-O, perform slightly better than the simple heuristics on the PolitiFact network, but show almost no advantage on the GossipCop network.
In contrast, the MIA-NPP consistently achieves superior suppression effect across both networks, demonstrating its effectiveness in identifying critical targets for prebunking.
Similar results are observed in other social network datasets as well (see Appendix~\ref{sect:results_for_other_datasets}).

To investigate the trends of selected prebunking targets, we visualized the first $100$ prebunking targets selected by each algorithm in a two-dimensional plane of the propagation probability $pp(P_{rv}^\ast)$ from root node $r$ and the misinformation susceptibility $q_v$ in the GossipCop network (Fig.~\ref{fig:intervention_target_map}).
From this figure, it can be seen that MIA-NPP prioritizes nodes that are both likely to receive misinformation from the root node and highly susceptible to it.

\begin{figure}[t]
    \centering
    \includegraphics[width=1.0\linewidth]{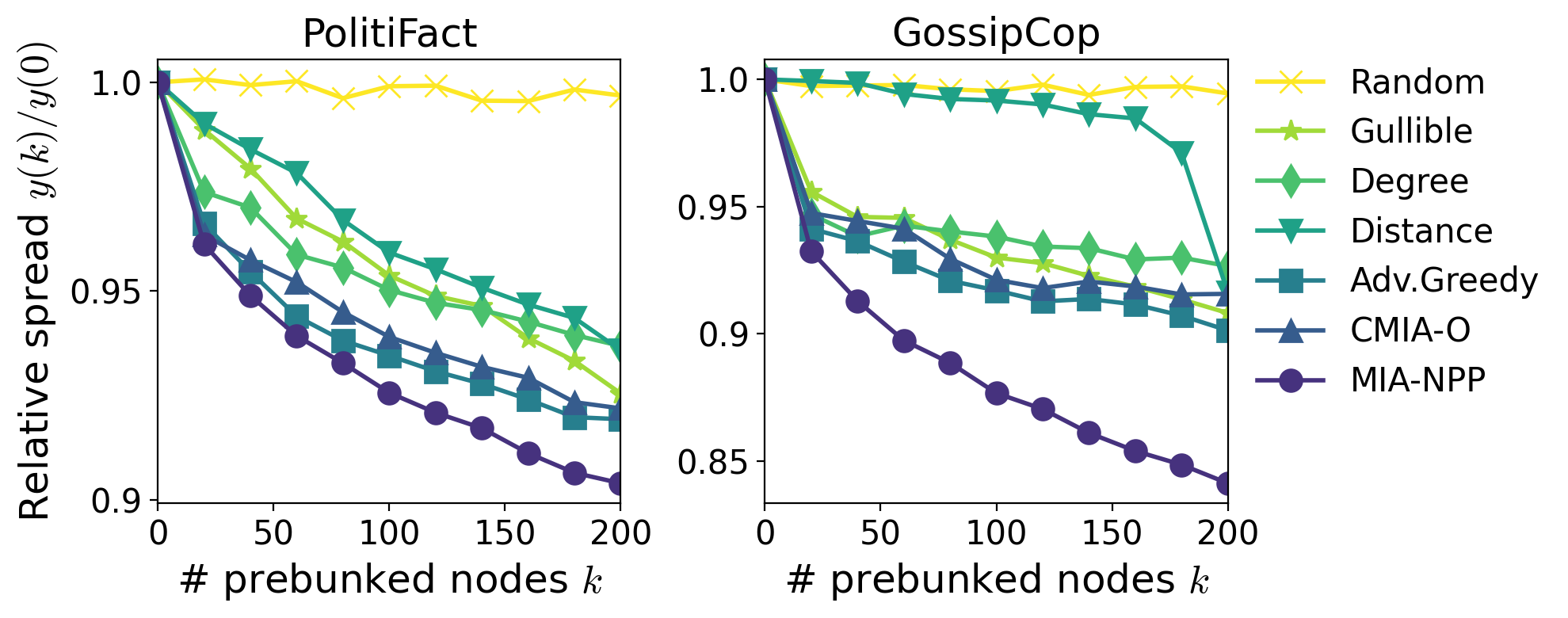}
    \caption{Misinformation suppression effects of each algorithm on the PolitiFact and GossipCop networks.}
    \label{fig:results_real}
\end{figure}

\begin{figure}[t]
    \centering
    \includegraphics[width=1.0\linewidth]{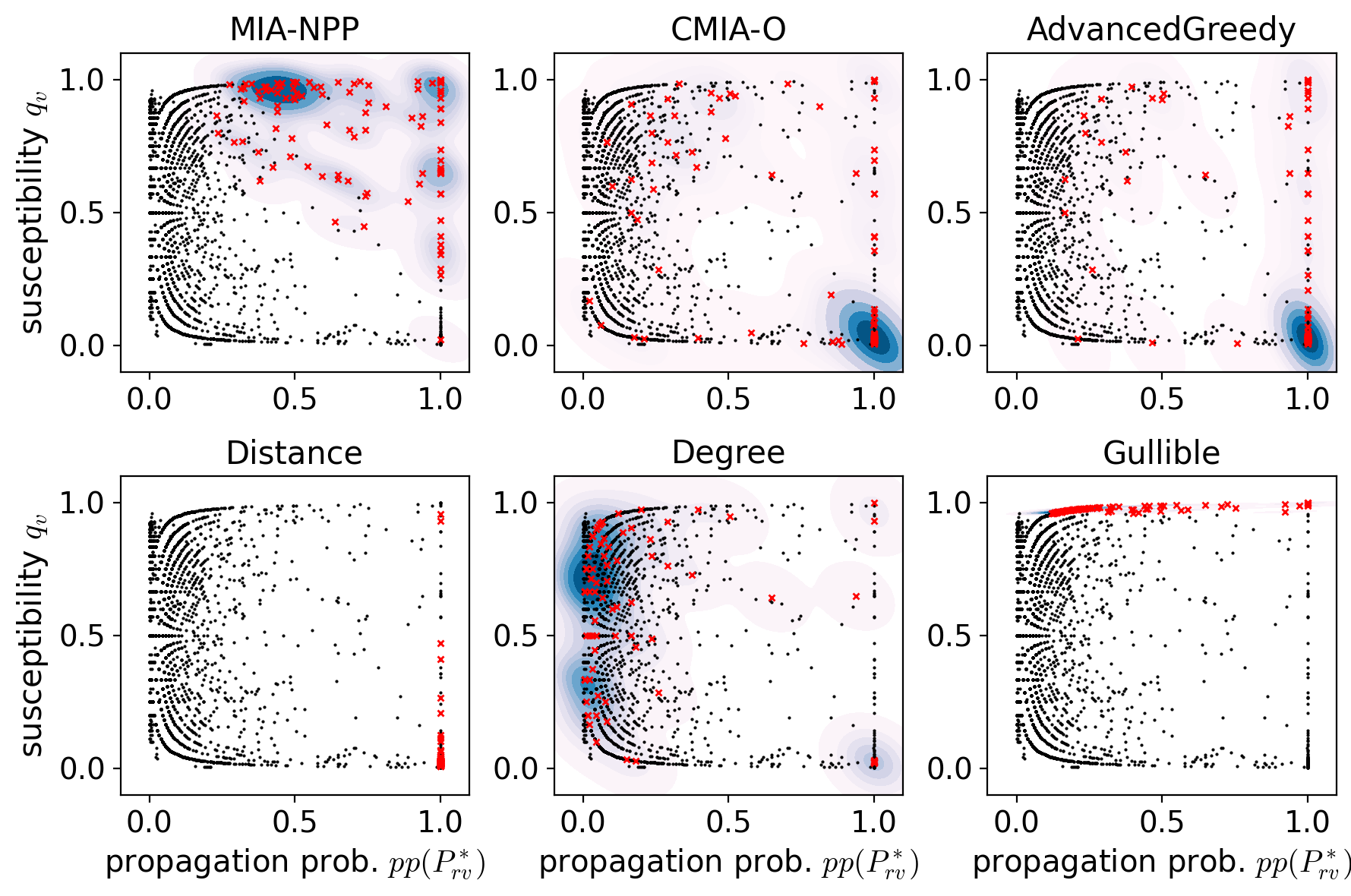}
    \caption{Visualizations of the first 100 prebunking targets selected by each algorithm on the GossipCop network, plotted on a 2D plane of $pp(P_{rv}^\ast)$ and $q_v$. Nodes are shown as black dots, with selected targets marked as red Xs.}
    \label{fig:intervention_target_map}
\end{figure}

\subsubsection{Evaluation under uncertain observations (RQ2)} 
\label{sect:eval_uncertain}

In Section~\ref{sect:eval_perfect}, we assumed that the true values of $q_v$ and $\varepsilon_v$ were fully observable when selecting prebunking targets using MIA-NPP.
However, in real-world settings, it is rare to observe the true susceptibility of users with perfect accuracy, and such measurements often involve noise.
Moreover, it is generally extremely difficult to estimate the individual intervention effect for each user.
Therefore, in this section, we evaluate the extent to which MIA-NPP remains effective under conditions where these parameters are uncertain.

To model uncertainty, we assume that the observed susceptibility $q_v^{\mathrm{obs}}$ includes observational noise $\delta \sim \mathcal{N}(0, \sigma_\delta^2)$; that is, $q_v^{\mathrm{obs}} = q_v + \delta$.
Additionally, we assume that the individual intervention effect $\varepsilon_v$ for each node is not directly observable, and instead only the average intervention effect is available; that is, $\varepsilon_v^\mathrm{obs} = \langle \varepsilon \rangle = 0.5$ for all $v \in V$.
Under this setting, we evaluate the effectiveness of MIA-NPP in suppressing misinformation while varying the noise variance $\sigma_\delta^2 \in \{0.0, 0.1, 0.5, 1.0\}$.

The evaluation results are shown in Fig.~\ref{fig:results_uncertain}.
As expected, the misinformation suppression effect of MIA-NPP decreases as the observational noise increases.
However, when the observational noise is sufficiently small, MIA-NPP still maintains a higher suppression effect than AdvancedGreedy (yellow dashed line), which achieves second-best performance in Fig.~\ref{fig:results_real}.
In addition, the results for $\sigma_\delta^2=0.0$ are almost identical to those of MIA-NPP (red dashed line) under perfect observation.
This suggests that even if only the average intervention effect $\langle \varepsilon \rangle$ is available instead of the individual intervention effect $\varepsilon_v$ in the selection of prebunking targets, it has little effect on the suppression performance of MIA-NPP.

\begin{figure}[t]
    \centering
    \includegraphics[width=1.0\linewidth]{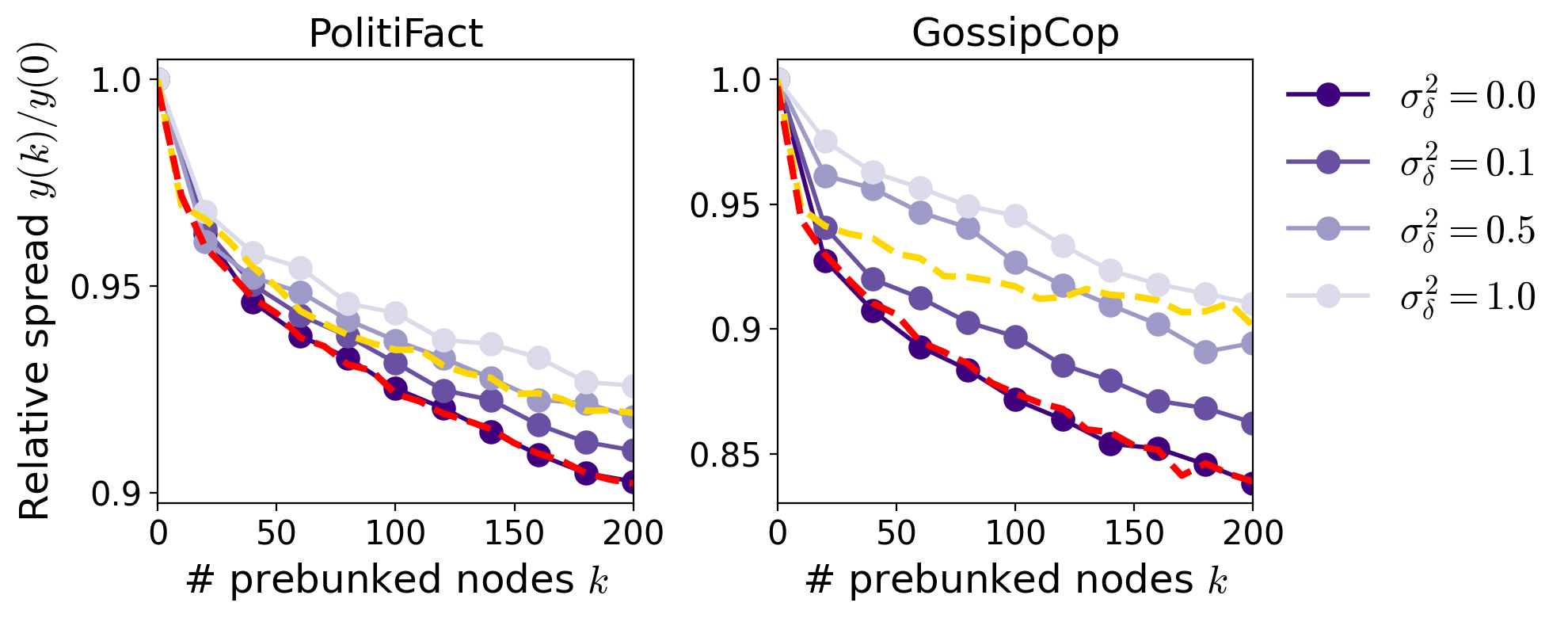}
    \caption{Misinformation suppression effects of MIA-NPP under uncertain observations on the PolitiFact and GossipCop networks.
    Red and yellow dashed lines respectively indicate the results of MIA-NPP and AdvancedGreedy under perfect observations, for comparison.}
    \label{fig:results_uncertain}
\end{figure}

\section{Conclusion}

In this study, we formulated a combinatorial optimization problem, called the network prebunking problem, to find the optimal set of prebunking targets that minimizes the spread of misinformation in a social network under an information diffusion model in which users who receive misinformation choose their actions probabilistically according to their susceptibility to misinformation.
We proved that the problem is NP-hard and that its objective function is monotone and submodular, which provides a theoretical foundation for approximation guarantees of greedy algorithms. 
However, since the greedy approach is computationally expensive and does not scale to large networks, we proposed an efficient approximation algorithm, MIA-NPP, based on the MIA framework, which restricts influence propagation to local arborescences rooted at each node. 
Furthermore, numerical experiments using real-world network datasets evaluated the effectiveness of MIA-NPP in both cases where misinformation susceptibility and individual intervention effects were fully observable and where they were only partially observable.
The results revealed that MIA-NPP can effectively suppress the spread of misinformation compared to simple heuristics based on user characteristics or network structure, as well as approximation algorithms for related problems.

Importantly, our work is the first to systematically investigate prebunking as an alternative to conventional intervention strategies for mitigating the spread of misinformation in social networks, such as blocking and clarification.
Although prebunking may have a relatively modest suppression effect because it does not prohibit intervened users from sharing misinformation (see Appendix~\ref{sect:compare_PBC}), it offers clear advantages in terms of ethical acceptability and implementation feasibility.
Our results suggest that even under limited observability and intervention resources, strategically designed prebunking can serve as a robust and realistic approach to counter misinformation.





\bibliographystyle{ACM-Reference-Format}
\bibliography{ref}


\begin{thebibliography}{61}


\ifx \showCODEN    \undefined \def \showCODEN     #1{\unskip}     \fi
\ifx \showDOI      \undefined \def \showDOI       #1{#1}\fi
\ifx \showISBNx    \undefined \def \showISBNx     #1{\unskip}     \fi
\ifx \showISBNxiii \undefined \def \showISBNxiii  #1{\unskip}     \fi
\ifx \showISSN     \undefined \def \showISSN      #1{\unskip}     \fi
\ifx \showLCCN     \undefined \def \showLCCN      #1{\unskip}     \fi
\ifx \shownote     \undefined \def \shownote      #1{#1}          \fi
\ifx \showarticletitle \undefined \def \showarticletitle #1{#1}   \fi
\ifx \showURL      \undefined \def \showURL       {\relax}        \fi
\providecommand\bibfield[2]{#2}
\providecommand\bibinfo[2]{#2}
\providecommand\natexlab[1]{#1}
\providecommand\showeprint[2][]{arXiv:#2}

\bibitem[Arazkhani et~al\mbox{.}(2019)]%
        {arazkhani2019influence}
\bibfield{author}{\bibinfo{person}{Niloofar Arazkhani},
  \bibinfo{person}{Mohammad~Reza Meybodi}, {and} \bibinfo{person}{Alireza
  Rezvanian}.} \bibinfo{year}{2019}\natexlab{}.
\newblock \showarticletitle{Influence blocking maximization in social network
  using centrality measures}. In \bibinfo{booktitle}{\emph{2019 5th Conference
  on Knowledge Based Engineering and Innovation (KBEI)}}. IEEE,
  \bibinfo{pages}{492--497}.
\newblock


\bibitem[Barman(2025)]%
        {barman2025rethinking}
\bibfield{author}{\bibinfo{person}{Dipto Barman}.}
  \bibinfo{year}{2025}\natexlab{}.
\newblock \showarticletitle{Rethinking Misinformation Mitigation: The Case for
  Personalized Digital Interventions}.
\newblock \bibinfo{journal}{\emph{ACM SIGWEB Newsletter}}
  \bibinfo{volume}{2025}, \bibinfo{number}{Winter} (\bibinfo{year}{2025}),
  \bibinfo{pages}{1--7}.
\newblock


\bibitem[Barman et~al\mbox{.}(2024)]%
        {barman2024personalised}
\bibfield{author}{\bibinfo{person}{Dipto Barman}, \bibinfo{person}{Kevin
  Koidl}, \bibinfo{person}{Sylvia~Magdalena Leiter},
  \bibinfo{person}{Ois{\'\i}n Carroll}, \bibinfo{person}{Alexander Nussbaumer},
  \bibinfo{person}{Brendan Spillane}, \bibinfo{person}{Owen Conlan}, {and}
  \bibinfo{person}{Gary Munnelly}.} \bibinfo{year}{2024}\natexlab{}.
\newblock \showarticletitle{A Personalised Adaptive Intervention Framework to
  Tackle Misinformation and Disinformation on Social Media}. In
  \bibinfo{booktitle}{\emph{2024 11th International Conference on Social
  Networks Analysis, Management and Security (SNAMS)}}. IEEE,
  \bibinfo{pages}{58--65}.
\newblock


\bibitem[Basol et~al\mbox{.}(2021)]%
        {basol2021towards}
\bibfield{author}{\bibinfo{person}{Melisa Basol}, \bibinfo{person}{Jon
  Roozenbeek}, \bibinfo{person}{Manon Berriche}, \bibinfo{person}{Fatih Uenal},
  \bibinfo{person}{William~P McClanahan}, {and} \bibinfo{person}{Sander Van~der
  Linden}.} \bibinfo{year}{2021}\natexlab{}.
\newblock \showarticletitle{Towards psychological herd immunity: Cross-cultural
  evidence for two prebunking interventions against COVID-19 misinformation}.
\newblock \bibinfo{journal}{\emph{Big Data \& Society}} \bibinfo{volume}{8},
  \bibinfo{number}{1} (\bibinfo{year}{2021}),
  \bibinfo{pages}{20539517211013868}.
\newblock


\bibitem[Basol et~al\mbox{.}(2020)]%
        {basol2020good}
\bibfield{author}{\bibinfo{person}{Melisa Basol}, \bibinfo{person}{Jon
  Roozenbeek}, {and} \bibinfo{person}{Sander Van~der Linden}.}
  \bibinfo{year}{2020}\natexlab{}.
\newblock \showarticletitle{Good news about bad news: Gamified inoculation
  boosts confidence and cognitive immunity against fake news}.
\newblock \bibinfo{journal}{\emph{Journal of cognition}} \bibinfo{volume}{3},
  \bibinfo{number}{1} (\bibinfo{year}{2020}), \bibinfo{pages}{2}.
\newblock


\bibitem[Bayiz and Topcu(2023)]%
        {bayiz2023prebunking}
\bibfield{author}{\bibinfo{person}{Yigit~Ege Bayiz} {and} \bibinfo{person}{Ufuk
  Topcu}.} \bibinfo{year}{2023}\natexlab{}.
\newblock \showarticletitle{Prebunking Design as a Defense Mechanism Against
  Misinformation Propagation on Social Networks}.
\newblock \bibinfo{journal}{\emph{arXiv preprint arXiv:2311.14200}}
  (\bibinfo{year}{2023}).
\newblock


\bibitem[Budak et~al\mbox{.}(2011)]%
        {budak2011limiting}
\bibfield{author}{\bibinfo{person}{Ceren Budak}, \bibinfo{person}{Divyakant
  Agrawal}, {and} \bibinfo{person}{Amr El~Abbadi}.}
  \bibinfo{year}{2011}\natexlab{}.
\newblock \showarticletitle{Limiting the spread of misinformation in social
  networks}. In \bibinfo{booktitle}{\emph{Proceedings of the 20th international
  conference on World wide web}}. \bibinfo{pages}{665--674}.
\newblock


\bibitem[Capewell et~al\mbox{.}(2024)]%
        {capewell2024misinformation}
\bibfield{author}{\bibinfo{person}{Georgia Capewell}, \bibinfo{person}{Rakoen
  Maertens}, \bibinfo{person}{Miriam Remshard}, \bibinfo{person}{Sander Van
  Der~Linden}, \bibinfo{person}{Josh Compton}, \bibinfo{person}{Stephan
  Lewandowsky}, {and} \bibinfo{person}{Jon Roozenbeek}.}
  \bibinfo{year}{2024}\natexlab{}.
\newblock \showarticletitle{Misinformation interventions decay rapidly without
  an immediate posttest}.
\newblock \bibinfo{journal}{\emph{Journal of Applied Social Psychology}}
  \bibinfo{volume}{54}, \bibinfo{number}{8} (\bibinfo{year}{2024}),
  \bibinfo{pages}{441--454}.
\newblock


\bibitem[Chan et~al\mbox{.}(2017)]%
        {chan2017debunking}
\bibfield{author}{\bibinfo{person}{Man-pui~Sally Chan},
  \bibinfo{person}{Christopher~R Jones}, \bibinfo{person}{Kathleen
  Hall~Jamieson}, {and} \bibinfo{person}{Dolores Albarrac{\'\i}n}.}
  \bibinfo{year}{2017}\natexlab{}.
\newblock \showarticletitle{Debunking: A meta-analysis of the psychological
  efficacy of messages countering misinformation}.
\newblock \bibinfo{journal}{\emph{Psychological Science}} \bibinfo{volume}{28},
  \bibinfo{number}{11} (\bibinfo{year}{2017}), \bibinfo{pages}{1531--1546}.
\newblock


\bibitem[Chen et~al\mbox{.}(2022)]%
        {chen2022influence}
\bibfield{author}{\bibinfo{person}{Bo-Lun Chen}, \bibinfo{person}{Wen-Xin
  Jiang}, \bibinfo{person}{Yi-Xin Chen}, \bibinfo{person}{Ling Chen},
  \bibinfo{person}{Rui-Jie Wang}, \bibinfo{person}{Shuai Han},
  \bibinfo{person}{Jian-Hong Lin}, {and} \bibinfo{person}{Yi-Cheng Zhang}.}
  \bibinfo{year}{2022}\natexlab{}.
\newblock \showarticletitle{Influence blocking maximization on networks:
  Models, methods and applications}.
\newblock \bibinfo{journal}{\emph{Physics Reports}}  \bibinfo{volume}{976}
  (\bibinfo{year}{2022}), \bibinfo{pages}{1--54}.
\newblock


\bibitem[Chen et~al\mbox{.}(2011)]%
        {chen2011influence}
\bibfield{author}{\bibinfo{person}{Wei Chen}, \bibinfo{person}{Alex Collins},
  \bibinfo{person}{Rachel Cummings}, \bibinfo{person}{Te Ke},
  \bibinfo{person}{Zhenming Liu}, \bibinfo{person}{David Rincon},
  \bibinfo{person}{Xiaorui Sun}, \bibinfo{person}{Yajun Wang},
  \bibinfo{person}{Wei Wei}, {and} \bibinfo{person}{Yifei Yuan}.}
  \bibinfo{year}{2011}\natexlab{}.
\newblock \showarticletitle{Influence maximization in social networks when
  negative opinions may emerge and propagate}. In
  \bibinfo{booktitle}{\emph{Proceedings of the 2011 SIAM International
  Conference on Data Mining}}. SIAM, \bibinfo{pages}{379--390}.
\newblock


\bibitem[Chen et~al\mbox{.}(2023)]%
        {chen2023neural}
\bibfield{author}{\bibinfo{person}{Wenjie Chen}, \bibinfo{person}{Shengcai
  Liu}, \bibinfo{person}{Yew-Soon Ong}, {and} \bibinfo{person}{Ke Tang}.}
  \bibinfo{year}{2023}\natexlab{}.
\newblock \showarticletitle{Neural influence estimator: Towards real-time
  solutions to influence blocking maximization}.
\newblock \bibinfo{journal}{\emph{arXiv preprint arXiv:2308.14012}}
  (\bibinfo{year}{2023}).
\newblock


\bibitem[Chen et~al\mbox{.}(2010a)]%
        {chen2010scalable_KDD}
\bibfield{author}{\bibinfo{person}{Wei Chen}, \bibinfo{person}{Chi Wang}, {and}
  \bibinfo{person}{Yajun Wang}.} \bibinfo{year}{2010}\natexlab{a}.
\newblock \showarticletitle{Scalable influence maximization for prevalent viral
  marketing in large-scale social networks}. In
  \bibinfo{booktitle}{\emph{Proceedings of the 16th ACM SIGKDD International
  Conference on Knowledge Discovery and Data Mining}}.
  \bibinfo{publisher}{ACM}, \bibinfo{pages}{1029--1038}.
\newblock


\bibitem[Chen et~al\mbox{.}(2010b)]%
        {chen2010scalable_ICDM}
\bibfield{author}{\bibinfo{person}{Wei Chen}, \bibinfo{person}{Yifei Yuan},
  {and} \bibinfo{person}{Li Zhang}.} \bibinfo{year}{2010}\natexlab{b}.
\newblock \showarticletitle{Scalable influence maximization in social networks
  under the linear threshold model}. In \bibinfo{booktitle}{\emph{2010 IEEE
  International Conference on Data Mining}}. IEEE, \bibinfo{pages}{88--97}.
\newblock


\bibitem[Compton et~al\mbox{.}(2021)]%
        {compton2021inoculation}
\bibfield{author}{\bibinfo{person}{Josh Compton}, \bibinfo{person}{Sander Van
  Der~Linden}, \bibinfo{person}{John Cook}, {and} \bibinfo{person}{Melisa
  Basol}.} \bibinfo{year}{2021}\natexlab{}.
\newblock \showarticletitle{Inoculation theory in the post-truth era: Extant
  findings and new frontiers for contested science, misinformation, and
  conspiracy theories}.
\newblock \bibinfo{journal}{\emph{Social and Personality Psychology Compass}}
  \bibinfo{volume}{15}, \bibinfo{number}{6} (\bibinfo{year}{2021}),
  \bibinfo{pages}{e12602}.
\newblock


\bibitem[Cook et~al\mbox{.}(2017)]%
        {cook2017neutralizing}
\bibfield{author}{\bibinfo{person}{John Cook}, \bibinfo{person}{Stephan
  Lewandowsky}, {and} \bibinfo{person}{Ullrich~KH Ecker}.}
  \bibinfo{year}{2017}\natexlab{}.
\newblock \showarticletitle{Neutralizing misinformation through inoculation:
  Exposing misleading argumentation techniques reduces their influence}.
\newblock \bibinfo{journal}{\emph{PloS one}} \bibinfo{volume}{12},
  \bibinfo{number}{5} (\bibinfo{year}{2017}), \bibinfo{pages}{e0175799}.
\newblock


\bibitem[Dou et~al\mbox{.}(2021)]%
        {dou2021user}
\bibfield{author}{\bibinfo{person}{Yingtong Dou}, \bibinfo{person}{Kai Shu},
  \bibinfo{person}{Congying Xia}, \bibinfo{person}{Philip~S. Yu}, {and}
  \bibinfo{person}{Lichao Sun}.} \bibinfo{year}{2021}\natexlab{}.
\newblock \showarticletitle{User Preference-aware Fake News Detection}. In
  \bibinfo{booktitle}{\emph{Proceedings of the 44th International ACM SIGIR
  Conference on Research and Development in Information Retrieval}}.
\newblock


\bibitem[Fan et~al\mbox{.}(2013)]%
        {fan2013least}
\bibfield{author}{\bibinfo{person}{Lidan Fan}, \bibinfo{person}{Zaixin Lu},
  \bibinfo{person}{Weili Wu}, \bibinfo{person}{Bhavani Thuraisingham},
  \bibinfo{person}{Huan Ma}, {and} \bibinfo{person}{Yuanjun Bi}.}
  \bibinfo{year}{2013}\natexlab{}.
\newblock \showarticletitle{Least cost rumor blocking in social networks}. In
  \bibinfo{booktitle}{\emph{2013 IEEE 33rd International Conference on
  Distributed Computing Systems}}. IEEE, \bibinfo{pages}{540--549}.
\newblock


\bibitem[Harjani et~al\mbox{.}(2022)]%
        {harjani2022practical}
\bibfield{author}{\bibinfo{person}{T Harjani}, \bibinfo{person}{J Roozenbeek},
  \bibinfo{person}{M Biddlestone}, \bibinfo{person}{S van~der Linden},
  \bibinfo{person}{A Stuart}, \bibinfo{person}{M Iwahara}, \bibinfo{person}{B
  Piri}, \bibinfo{person}{R Xu}, \bibinfo{person}{B Goldberg}, {and}
  \bibinfo{person}{M Graham}.} \bibinfo{year}{2022}\natexlab{}.
\newblock \bibinfo{title}{A practical Guide to prebunking misinformation}.
\newblock
\newblock
\urldef\tempurl%
\url{https://prebunking.withgoogle.com/docs/A_Practical_Guide_to_Prebunking_Misinformation.pdf}
\showURL{%
\tempurl}
\newblock
\shownote{Accessed: July 27, 2025}.


\bibitem[He et~al\mbox{.}(2012)]%
        {he2012influence}
\bibfield{author}{\bibinfo{person}{Xinran He}, \bibinfo{person}{Guojie Song},
  \bibinfo{person}{Wei Chen}, {and} \bibinfo{person}{Qingye Jiang}.}
  \bibinfo{year}{2012}\natexlab{}.
\newblock \showarticletitle{Influence blocking maximization in social networks
  under the competitive linear threshold model}. In
  \bibinfo{booktitle}{\emph{Proceedings of the 2012 SIAM International
  Conference on Data Mining}}. SIAM, \bibinfo{pages}{463--474}.
\newblock


\bibitem[Hosni et~al\mbox{.}(2019)]%
        {hosni2019darim}
\bibfield{author}{\bibinfo{person}{Adil Imad~Eddine Hosni},
  \bibinfo{person}{Kan Li}, {and} \bibinfo{person}{Sadique Ahmad}.}
  \bibinfo{year}{2019}\natexlab{}.
\newblock \showarticletitle{DARIM: Dynamic approach for rumor influence
  minimization in online social networks}. In
  \bibinfo{booktitle}{\emph{International Conference on Neural Information
  Processing}}. Springer, \bibinfo{pages}{619--630}.
\newblock


\bibitem[Kempe et~al\mbox{.}(2003)]%
        {kempe2003maximizing}
\bibfield{author}{\bibinfo{person}{David Kempe}, \bibinfo{person}{Jon
  Kleinberg}, {and} \bibinfo{person}{{\'E}va Tardos}.}
  \bibinfo{year}{2003}\natexlab{}.
\newblock \showarticletitle{Maximizing the spread of influence through a social
  network}. In \bibinfo{booktitle}{\emph{Proceedings of the ninth ACM SIGKDD
  international conference on Knowledge discovery and data mining}}.
  \bibinfo{pages}{137--146}.
\newblock


\bibitem[Khalil et~al\mbox{.}(2013)]%
        {khalil2013cuttingedge}
\bibfield{author}{\bibinfo{person}{Elias Khalil}, \bibinfo{person}{Bistra
  Dilkina}, {and} \bibinfo{person}{Le Song}.} \bibinfo{year}{2013}\natexlab{}.
\newblock \showarticletitle{Cuttingedge: Influence minimization in networks}.
  In \bibinfo{booktitle}{\emph{Proceedings of Workshop on Frontiers of Network
  Analysis: Methods, Models, and Applications at NIPS}}.
  \bibinfo{pages}{1--13}.
\newblock


\bibitem[Kimura et~al\mbox{.}(2008)]%
        {kimura2008solving}
\bibfield{author}{\bibinfo{person}{Masahiro Kimura}, \bibinfo{person}{Kazumi
  Saito}, {and} \bibinfo{person}{Hiroshi Motoda}.}
  \bibinfo{year}{2008}\natexlab{}.
\newblock \showarticletitle{Solving the contamination minimization problem on
  networks for the linear threshold model}. In
  \bibinfo{booktitle}{\emph{Pacific rim international conference on artificial
  intelligence}}. Springer, \bibinfo{pages}{977--984}.
\newblock


\bibitem[Kimura et~al\mbox{.}(2009)]%
        {kimura2009blocking}
\bibfield{author}{\bibinfo{person}{Masahiro Kimura}, \bibinfo{person}{Kazumi
  Saito}, {and} \bibinfo{person}{Hiroshi Motoda}.}
  \bibinfo{year}{2009}\natexlab{}.
\newblock \showarticletitle{Blocking links to minimize contamination spread in
  a social network}.
\newblock \bibinfo{journal}{\emph{ACM Transactions on Knowledge Discovery from
  Data (TKDD)}} \bibinfo{volume}{3}, \bibinfo{number}{2}
  (\bibinfo{year}{2009}), \bibinfo{pages}{1--23}.
\newblock


\bibitem[Kozyreva et~al\mbox{.}(2024)]%
        {kozyreva2024toolbox}
\bibfield{author}{\bibinfo{person}{Anastasia Kozyreva},
  \bibinfo{person}{Philipp Lorenz-Spreen}, \bibinfo{person}{Stefan~M Herzog},
  \bibinfo{person}{Ullrich~KH Ecker}, \bibinfo{person}{Stephan Lewandowsky},
  \bibinfo{person}{Ralph Hertwig}, \bibinfo{person}{Ayesha Ali},
  \bibinfo{person}{Joe Bak-Coleman}, \bibinfo{person}{Sarit Barzilai},
  \bibinfo{person}{Melisa Basol}, {et~al\mbox{.}}}
  \bibinfo{year}{2024}\natexlab{}.
\newblock \showarticletitle{Toolbox of individual-level interventions against
  online misinformation}.
\newblock \bibinfo{journal}{\emph{Nature Human Behaviour}} \bibinfo{volume}{8},
  \bibinfo{number}{6} (\bibinfo{year}{2024}), \bibinfo{pages}{1044--1052}.
\newblock


\bibitem[Leskovec et~al\mbox{.}(2007)]%
        {leskovec2007cost}
\bibfield{author}{\bibinfo{person}{Jure Leskovec}, \bibinfo{person}{Andreas
  Krause}, \bibinfo{person}{Carlos Guestrin}, \bibinfo{person}{Christos
  Faloutsos}, \bibinfo{person}{Jeanne VanBriesen}, {and}
  \bibinfo{person}{Natalie Glance}.} \bibinfo{year}{2007}\natexlab{}.
\newblock \showarticletitle{Cost-effective outbreak detection in networks}. In
  \bibinfo{booktitle}{\emph{Proceedings of the 13th ACM SIGKDD international
  conference on Knowledge discovery and data mining}}.
  \bibinfo{pages}{420--429}.
\newblock


\bibitem[Lewandowsky et~al\mbox{.}(2012)]%
        {lewandowsky2012misinformation}
\bibfield{author}{\bibinfo{person}{Stephan Lewandowsky},
  \bibinfo{person}{Ullrich~KH Ecker}, \bibinfo{person}{Colleen~M Seifert},
  \bibinfo{person}{Norbert Schwarz}, {and} \bibinfo{person}{John Cook}.}
  \bibinfo{year}{2012}\natexlab{}.
\newblock \showarticletitle{Misinformation and its correction: Continued
  influence and successful debiasing}.
\newblock \bibinfo{journal}{\emph{Psychological science in the public
  interest}} \bibinfo{volume}{13}, \bibinfo{number}{3} (\bibinfo{year}{2012}),
  \bibinfo{pages}{106--131}.
\newblock


\bibitem[Lewandowsky and Van Der~Linden(2021)]%
        {lewandowsky2021countering}
\bibfield{author}{\bibinfo{person}{Stephan Lewandowsky} {and}
  \bibinfo{person}{Sander Van Der~Linden}.} \bibinfo{year}{2021}\natexlab{}.
\newblock \showarticletitle{Countering misinformation and fake news through
  inoculation and prebunking}.
\newblock \bibinfo{journal}{\emph{European Review of Social Psychology}}
  \bibinfo{volume}{32}, \bibinfo{number}{2} (\bibinfo{year}{2021}),
  \bibinfo{pages}{348--384}.
\newblock


\bibitem[Lv et~al\mbox{.}(2019)]%
        {lv2019community}
\bibfield{author}{\bibinfo{person}{Jiaguo Lv}, \bibinfo{person}{Bin Yang},
  \bibinfo{person}{Zhen Yang}, {and} \bibinfo{person}{Wei Zhang}.}
  \bibinfo{year}{2019}\natexlab{}.
\newblock \showarticletitle{A community-based algorithm for influence blocking
  maximization in social networks}.
\newblock \bibinfo{journal}{\emph{Cluster Computing}}  \bibinfo{volume}{22}
  (\bibinfo{year}{2019}), \bibinfo{pages}{5587--5602}.
\newblock


\bibitem[Maertens et~al\mbox{.}(2021)]%
        {maertens2021long}
\bibfield{author}{\bibinfo{person}{Rakoen Maertens}, \bibinfo{person}{Jon
  Roozenbeek}, \bibinfo{person}{Melisa Basol}, {and} \bibinfo{person}{Sander
  van~der Linden}.} \bibinfo{year}{2021}\natexlab{}.
\newblock \showarticletitle{Long-term effectiveness of inoculation against
  misinformation: Three longitudinal experiments.}
\newblock \bibinfo{journal}{\emph{Journal of Experimental Psychology: Applied}}
  \bibinfo{volume}{27}, \bibinfo{number}{1} (\bibinfo{year}{2021}),
  \bibinfo{pages}{1}.
\newblock


\bibitem[Martel and Rand(2024)]%
        {martel2024fact}
\bibfield{author}{\bibinfo{person}{Cameron Martel} {and}
  \bibinfo{person}{David~G Rand}.} \bibinfo{year}{2024}\natexlab{}.
\newblock \showarticletitle{Fact-checker warning labels are effective even for
  those who distrust fact-checkers}.
\newblock \bibinfo{journal}{\emph{Nature Human Behaviour}} \bibinfo{volume}{8},
  \bibinfo{number}{10} (\bibinfo{year}{2024}), \bibinfo{pages}{1957--1967}.
\newblock


\bibitem[McGuire and Papageorgis(1961)]%
        {mcguire1961relative}
\bibfield{author}{\bibinfo{person}{William~J McGuire} {and}
  \bibinfo{person}{Demetrios Papageorgis}.} \bibinfo{year}{1961}\natexlab{}.
\newblock \showarticletitle{The relative efficacy of various types of prior
  belief-defense in producing immunity against persuasion.}
\newblock \bibinfo{journal}{\emph{The Journal of Abnormal and Social
  Psychology}} \bibinfo{volume}{62}, \bibinfo{number}{2}
  (\bibinfo{year}{1961}), \bibinfo{pages}{327}.
\newblock


\bibitem[Nemhauser et~al\mbox{.}(1978)]%
        {nemhauser1978analysis}
\bibfield{author}{\bibinfo{person}{George~L Nemhauser},
  \bibinfo{person}{Laurence~A Wolsey}, {and} \bibinfo{person}{Marshall~L
  Fisher}.} \bibinfo{year}{1978}\natexlab{}.
\newblock \showarticletitle{An analysis of approximations for maximizing
  submodular set functions—I}.
\newblock \bibinfo{journal}{\emph{Mathematical Programming}}
  \bibinfo{volume}{14} (\bibinfo{year}{1978}), \bibinfo{pages}{265--294}.
\newblock


\bibitem[Paynter et~al\mbox{.}(2019)]%
        {paynter2019evaluation}
\bibfield{author}{\bibinfo{person}{Jessica Paynter}, \bibinfo{person}{Sarah
  Luskin-Saxby}, \bibinfo{person}{Deb Keen}, \bibinfo{person}{Kathryn Fordyce},
  \bibinfo{person}{Grace Frost}, \bibinfo{person}{Christine Imms},
  \bibinfo{person}{Scott Miller}, \bibinfo{person}{David Trembath},
  \bibinfo{person}{Madonna Tucker}, {and} \bibinfo{person}{Ullrich Ecker}.}
  \bibinfo{year}{2019}\natexlab{}.
\newblock \showarticletitle{Evaluation of a template for countering
  misinformation—Real-world Autism treatment myth debunking}.
\newblock \bibinfo{journal}{\emph{PloS one}} \bibinfo{volume}{14},
  \bibinfo{number}{1} (\bibinfo{year}{2019}), \bibinfo{pages}{e0210746}.
\newblock


\bibitem[Roozenbeek et~al\mbox{.}(2022a)]%
        {roozenbeek2022technique}
\bibfield{author}{\bibinfo{person}{Jon Roozenbeek}, \bibinfo{person}{Cecilie~S
  Traberg}, {and} \bibinfo{person}{Sander van~der Linden}.}
  \bibinfo{year}{2022}\natexlab{a}.
\newblock \showarticletitle{Technique-based inoculation against real-world
  misinformation}.
\newblock \bibinfo{journal}{\emph{Royal Society open science}}
  \bibinfo{volume}{9}, \bibinfo{number}{5} (\bibinfo{year}{2022}),
  \bibinfo{pages}{211719}.
\newblock


\bibitem[Roozenbeek and Van Der~Linden(2019)]%
        {roozenbeek2019the}
\bibfield{author}{\bibinfo{person}{Jon Roozenbeek} {and}
  \bibinfo{person}{Sander Van Der~Linden}.} \bibinfo{year}{2019}\natexlab{}.
\newblock \showarticletitle{The fake news game: actively inoculating against
  the risk of misinformation}.
\newblock \bibinfo{journal}{\emph{Journal of risk research}}
  \bibinfo{volume}{22}, \bibinfo{number}{5} (\bibinfo{year}{2019}),
  \bibinfo{pages}{570--580}.
\newblock


\bibitem[Roozenbeek and Van~der Linden(2019)]%
        {roozenbeek2019fake}
\bibfield{author}{\bibinfo{person}{Jon Roozenbeek} {and}
  \bibinfo{person}{Sander Van~der Linden}.} \bibinfo{year}{2019}\natexlab{}.
\newblock \showarticletitle{Fake news game confers psychological resistance
  against online misinformation}.
\newblock \bibinfo{journal}{\emph{Palgrave Communications}}
  \bibinfo{volume}{5}, \bibinfo{number}{1} (\bibinfo{year}{2019}),
  \bibinfo{pages}{1--10}.
\newblock


\bibitem[Roozenbeek and van~der Linden(2020)]%
        {roozenbeek2020breaking}
\bibfield{author}{\bibinfo{person}{Jon Roozenbeek} {and}
  \bibinfo{person}{Sander van~der Linden}.} \bibinfo{year}{2020}\natexlab{}.
\newblock \showarticletitle{Breaking Harmony Square: A game that ``inoculates''
  against political misinformation}.
\newblock  (\bibinfo{year}{2020}).
\newblock


\bibitem[Roozenbeek et~al\mbox{.}(2022b)]%
        {roozenbeek2022psychological}
\bibfield{author}{\bibinfo{person}{Jon Roozenbeek}, \bibinfo{person}{Sander Van
  Der~Linden}, \bibinfo{person}{Beth Goldberg}, \bibinfo{person}{Steve Rathje},
  {and} \bibinfo{person}{Stephan Lewandowsky}.}
  \bibinfo{year}{2022}\natexlab{b}.
\newblock \showarticletitle{Psychological inoculation improves resilience
  against misinformation on social media}.
\newblock \bibinfo{journal}{\emph{Science advances}} \bibinfo{volume}{8},
  \bibinfo{number}{34} (\bibinfo{year}{2022}), \bibinfo{pages}{eabo6254}.
\newblock


\bibitem[Roozenbeek et~al\mbox{.}(2020)]%
        {roozenbeek2020prebunking}
\bibfield{author}{\bibinfo{person}{Jon Roozenbeek}, \bibinfo{person}{Sander Van
  Der~Linden}, {and} \bibinfo{person}{Thomas Nygren}.}
  \bibinfo{year}{2020}\natexlab{}.
\newblock \showarticletitle{Prebunking interventions based on ``inoculation''
  theory can reduce susceptibility to misinformation across cultures}.
\newblock \bibinfo{journal}{\emph{Harvard Kennedy School (HKS) Misinformation
  Review}} (\bibinfo{year}{2020}).
\newblock


\bibitem[Saeed et~al\mbox{.}(2022)]%
        {saeed2022crowdsourced}
\bibfield{author}{\bibinfo{person}{Mohammed Saeed}, \bibinfo{person}{Nicolas
  Traub}, \bibinfo{person}{Maelle Nicolas}, \bibinfo{person}{Gianluca
  Demartini}, {and} \bibinfo{person}{Paolo Papotti}.}
  \bibinfo{year}{2022}\natexlab{}.
\newblock \showarticletitle{Crowdsourced fact-checking at Twitter: How does the
  crowd compare with experts?}. In \bibinfo{booktitle}{\emph{Proceedings of the
  31st ACM international conference on information \& knowledge management}}.
  \bibinfo{pages}{1736--1746}.
\newblock


\bibitem[Shu et~al\mbox{.}(2018)]%
        {shu2018fakenewsnet}
\bibfield{author}{\bibinfo{person}{Kai Shu}, \bibinfo{person}{Deepak
  Mahudeswaran}, \bibinfo{person}{Suhang Wang}, \bibinfo{person}{Dongwon Lee},
  {and} \bibinfo{person}{Huan Liu}.} \bibinfo{year}{2018}\natexlab{}.
\newblock \showarticletitle{FakeNewsNet: A Data Repository with News Content,
  Social Context and Dynamic Information for Studying Fake News on Social
  Media}.
\newblock \bibinfo{journal}{\emph{arXiv preprint arXiv:1809.01286}}
  (\bibinfo{year}{2018}).
\newblock


\bibitem[Smith et~al\mbox{.}(2011)]%
        {smith2011correcting}
\bibfield{author}{\bibinfo{person}{Philip Smith}, \bibinfo{person}{Maansi
  Bansal-Travers}, \bibinfo{person}{Richard O'Connor}, \bibinfo{person}{Anthony
  Brown}, \bibinfo{person}{Chris Banthin}, \bibinfo{person}{Sara
  Guardino-Colket}, {and} \bibinfo{person}{K~Michael Cummings}.}
  \bibinfo{year}{2011}\natexlab{}.
\newblock \showarticletitle{Correcting over 50 years of tobacco industry
  misinformation}.
\newblock \bibinfo{journal}{\emph{American journal of Preventive Medicine}}
  \bibinfo{volume}{40}, \bibinfo{number}{6} (\bibinfo{year}{2011}),
  \bibinfo{pages}{690--698}.
\newblock


\bibitem[Tay et~al\mbox{.}(2022)]%
        {tay2022comparison}
\bibfield{author}{\bibinfo{person}{Li~Qian Tay}, \bibinfo{person}{Mark~J
  Hurlstone}, \bibinfo{person}{Tim Kurz}, {and} \bibinfo{person}{Ullrich~KH
  Ecker}.} \bibinfo{year}{2022}\natexlab{}.
\newblock \showarticletitle{A comparison of prebunking and debunking
  interventions for implied versus explicit misinformation}.
\newblock \bibinfo{journal}{\emph{British Journal of Psychology}}
  \bibinfo{volume}{113}, \bibinfo{number}{3} (\bibinfo{year}{2022}),
  \bibinfo{pages}{591--607}.
\newblock


\bibitem[Tong(2020)]%
        {tong2020stratlearner}
\bibfield{author}{\bibinfo{person}{Guangmo Tong}.}
  \bibinfo{year}{2020}\natexlab{}.
\newblock \showarticletitle{StratLearner: Learning a strategy for
  misinformation prevention in social networks}.
\newblock \bibinfo{journal}{\emph{Advances in Neural Information Processing
  Systems}}  \bibinfo{volume}{33} (\bibinfo{year}{2020}),
  \bibinfo{pages}{15546--15555}.
\newblock


\bibitem[Traberg et~al\mbox{.}(2022)]%
        {traberg2022psychological}
\bibfield{author}{\bibinfo{person}{Cecilie~S Traberg}, \bibinfo{person}{Jon
  Roozenbeek}, {and} \bibinfo{person}{Sander Van Der~Linden}.}
  \bibinfo{year}{2022}\natexlab{}.
\newblock \showarticletitle{Psychological inoculation against misinformation:
  Current evidence and future directions}.
\newblock \bibinfo{journal}{\emph{The ANNALS of the American Academy of
  Political and Social Science}} \bibinfo{volume}{700}, \bibinfo{number}{1}
  (\bibinfo{year}{2022}), \bibinfo{pages}{136--151}.
\newblock


\bibitem[Van~der Linden et~al\mbox{.}(2017)]%
        {van2017inoculating}
\bibfield{author}{\bibinfo{person}{Sander Van~der Linden},
  \bibinfo{person}{Anthony Leiserowitz}, \bibinfo{person}{Seth Rosenthal},
  {and} \bibinfo{person}{Edward Maibach}.} \bibinfo{year}{2017}\natexlab{}.
\newblock \showarticletitle{Inoculating the public against misinformation about
  climate change}.
\newblock \bibinfo{journal}{\emph{Global Challenges}} \bibinfo{volume}{1},
  \bibinfo{number}{2} (\bibinfo{year}{2017}), \bibinfo{pages}{1600008}.
\newblock


\bibitem[Walter and Murphy(2018)]%
        {walter2018unring}
\bibfield{author}{\bibinfo{person}{Nathan Walter} {and}
  \bibinfo{person}{Sheila~T Murphy}.} \bibinfo{year}{2018}\natexlab{}.
\newblock \showarticletitle{How to unring the bell: A meta-analytic approach to
  correction of misinformation}.
\newblock \bibinfo{journal}{\emph{Communication Monographs}}
  \bibinfo{volume}{85}, \bibinfo{number}{3} (\bibinfo{year}{2018}),
  \bibinfo{pages}{423--441}.
\newblock


\bibitem[Wang et~al\mbox{.}(2017)]%
        {wang2017drimux}
\bibfield{author}{\bibinfo{person}{Biao Wang}, \bibinfo{person}{Ge Chen},
  \bibinfo{person}{Luoyi Fu}, \bibinfo{person}{Li Song}, {and}
  \bibinfo{person}{Xinbing Wang}.} \bibinfo{year}{2017}\natexlab{}.
\newblock \showarticletitle{Drimux: Dynamic rumor influence minimization with
  user experience in social networks}.
\newblock \bibinfo{journal}{\emph{IEEE Transactions on Knowledge and Data
  Engineering}} \bibinfo{volume}{29}, \bibinfo{number}{10}
  (\bibinfo{year}{2017}), \bibinfo{pages}{2168--2181}.
\newblock


\bibitem[Wang et~al\mbox{.}(2013)]%
        {wang2013negative}
\bibfield{author}{\bibinfo{person}{Senzhang Wang}, \bibinfo{person}{Xiaojian
  Zhao}, \bibinfo{person}{Yan Chen}, \bibinfo{person}{Zhoujun Li},
  \bibinfo{person}{Kai Zhang}, {and} \bibinfo{person}{Jiali Xia}.}
  \bibinfo{year}{2013}\natexlab{}.
\newblock \showarticletitle{Negative influence minimizing by blocking nodes in
  social networks}. In \bibinfo{booktitle}{\emph{Proceedings of the 17th AAAI
  Conference on Late-Breaking Developments in the Field of Artificial
  Intelligence}}. \bibinfo{pages}{134--136}.
\newblock


\bibitem[Wen et~al\mbox{.}(2014)]%
        {wen2014shut}
\bibfield{author}{\bibinfo{person}{Sheng Wen}, \bibinfo{person}{Jiaojiao
  Jiang}, \bibinfo{person}{Yang Xiang}, \bibinfo{person}{Shui Yu},
  \bibinfo{person}{Wanlei Zhou}, {and} \bibinfo{person}{Weijia Jia}.}
  \bibinfo{year}{2014}\natexlab{}.
\newblock \showarticletitle{To shut them up or to clarify: Restraining the
  spread of rumors in online social networks}.
\newblock \bibinfo{journal}{\emph{IEEE Transactions on Parallel and Distributed
  Systems}} \bibinfo{volume}{25}, \bibinfo{number}{12} (\bibinfo{year}{2014}),
  \bibinfo{pages}{3306--3316}.
\newblock


\bibitem[Wojcik et~al\mbox{.}(2022)]%
        {wojcik2022birdwatch}
\bibfield{author}{\bibinfo{person}{Stefan Wojcik}, \bibinfo{person}{Sophie
  Hilgard}, \bibinfo{person}{Nick Judd}, \bibinfo{person}{Delia Mocanu},
  \bibinfo{person}{Stephen Ragain}, \bibinfo{person}{MB Hunzaker},
  \bibinfo{person}{Keith Coleman}, {and} \bibinfo{person}{Jay Baxter}.}
  \bibinfo{year}{2022}\natexlab{}.
\newblock \showarticletitle{Birdwatch: Crowd wisdom and bridging algorithms can
  inform understanding and reduce the spread of misinformation}.
\newblock \bibinfo{journal}{\emph{arXiv preprint arXiv:2210.15723}}
  (\bibinfo{year}{2022}).
\newblock


\bibitem[Wu and Pan(2017)]%
        {wu2017scalable}
\bibfield{author}{\bibinfo{person}{Peng Wu} {and} \bibinfo{person}{Li Pan}.}
  \bibinfo{year}{2017}\natexlab{}.
\newblock \showarticletitle{Scalable influence blocking maximization in social
  networks under competitive independent cascade models}.
\newblock \bibinfo{journal}{\emph{Computer Networks}}  \bibinfo{volume}{123}
  (\bibinfo{year}{2017}), \bibinfo{pages}{38--50}.
\newblock


\bibitem[Xie et~al\mbox{.}(2023)]%
        {xie2023minimizing}
\bibfield{author}{\bibinfo{person}{Jiadong Xie}, \bibinfo{person}{Fan Zhang},
  \bibinfo{person}{Kai Wang}, \bibinfo{person}{Xuemin Lin}, {and}
  \bibinfo{person}{Wenjie Zhang}.} \bibinfo{year}{2023}\natexlab{}.
\newblock \showarticletitle{Minimizing the influence of misinformation via
  vertex blocking}. In \bibinfo{booktitle}{\emph{2023 IEEE 39th International
  Conference on Data Engineering (ICDE)}}. IEEE, \bibinfo{pages}{789--801}.
\newblock


\bibitem[Yan et~al\mbox{.}(2019)]%
        {yan2019minimizing}
\bibfield{author}{\bibinfo{person}{Ruidong Yan}, \bibinfo{person}{Deying Li},
  \bibinfo{person}{Weili Wu}, \bibinfo{person}{Ding-Zhu Du}, {and}
  \bibinfo{person}{Yongcai Wang}.} \bibinfo{year}{2019}\natexlab{}.
\newblock \showarticletitle{Minimizing influence of rumors by blockers on
  social networks: algorithms and analysis}.
\newblock \bibinfo{journal}{\emph{IEEE Transactions on Network Science and
  Engineering}} \bibinfo{volume}{7}, \bibinfo{number}{3}
  (\bibinfo{year}{2019}), \bibinfo{pages}{1067--1078}.
\newblock


\bibitem[Yao et~al\mbox{.}(2015)]%
        {yao2015topic}
\bibfield{author}{\bibinfo{person}{Qipeng Yao}, \bibinfo{person}{Ruisheng Shi},
  \bibinfo{person}{Chuan Zhou}, \bibinfo{person}{Peng Wang}, {and}
  \bibinfo{person}{Li Guo}.} \bibinfo{year}{2015}\natexlab{}.
\newblock \showarticletitle{Topic-aware social influence minimization}. In
  \bibinfo{booktitle}{\emph{Proceedings of the 24th International Conference on
  World Wide Web}}. \bibinfo{pages}{139--140}.
\newblock


\bibitem[Yao et~al\mbox{.}(2014)]%
        {yao2014minimizing}
\bibfield{author}{\bibinfo{person}{Qipeng Yao}, \bibinfo{person}{Chuan Zhou},
  \bibinfo{person}{Linbo Xiang}, \bibinfo{person}{Yanan Cao}, {and}
  \bibinfo{person}{Li Guo}.} \bibinfo{year}{2014}\natexlab{}.
\newblock \showarticletitle{Minimizing the negative influence by blocking links
  in social networks}. In \bibinfo{booktitle}{\emph{International conference on
  trustworthy computing and services}}. Springer, \bibinfo{pages}{65--73}.
\newblock


\bibitem[Yousuf et~al\mbox{.}(2021)]%
        {yousuf2021media}
\bibfield{author}{\bibinfo{person}{Hamza Yousuf}, \bibinfo{person}{Sander
  van~der Linden}, \bibinfo{person}{Luke Bredius}, \bibinfo{person}{GA~Ted van
  Essen}, \bibinfo{person}{Govert Sweep}, \bibinfo{person}{Zohar Preminger},
  \bibinfo{person}{Eric van Gorp}, \bibinfo{person}{Erik Scherder},
  \bibinfo{person}{Jagat Narula}, {and} \bibinfo{person}{Leonard Hofstra}.}
  \bibinfo{year}{2021}\natexlab{}.
\newblock \showarticletitle{A media intervention applying debunking versus
  non-debunking content to combat vaccine misinformation in elderly in the
  Netherlands: A digital randomised trial}.
\newblock \bibinfo{journal}{\emph{EClinicalMedicine}}  \bibinfo{volume}{35}
  (\bibinfo{year}{2021}).
\newblock


\bibitem[Zareie and Sakellariou(2021)]%
        {zareie2021minimizing}
\bibfield{author}{\bibinfo{person}{Ahmad Zareie} {and} \bibinfo{person}{Rizos
  Sakellariou}.} \bibinfo{year}{2021}\natexlab{}.
\newblock \showarticletitle{Minimizing the spread of misinformation in online
  social networks: A survey}.
\newblock \bibinfo{journal}{\emph{Journal of Network and Computer
  Applications}}  \bibinfo{volume}{186} (\bibinfo{year}{2021}),
  \bibinfo{pages}{103094}.
\newblock


\bibitem[Zollo et~al\mbox{.}(2017)]%
        {zollo2017debunking}
\bibfield{author}{\bibinfo{person}{Fabiana Zollo}, \bibinfo{person}{Alessandro
  Bessi}, \bibinfo{person}{Michela Del~Vicario}, \bibinfo{person}{Antonio
  Scala}, \bibinfo{person}{Guido Caldarelli}, \bibinfo{person}{Louis
  Shekhtman}, \bibinfo{person}{Shlomo Havlin}, {and} \bibinfo{person}{Walter
  Quattrociocchi}.} \bibinfo{year}{2017}\natexlab{}.
\newblock \showarticletitle{Debunking in a world of tribes}.
\newblock \bibinfo{journal}{\emph{PloS one}} \bibinfo{volume}{12},
  \bibinfo{number}{7} (\bibinfo{year}{2017}), \bibinfo{pages}{e0181821}.
\newblock


\end{thebibliography}

\appendix

\section{Notations}
\label{appendix:notation}
For clarity, Table~\ref{tab:notation} summarizes the main symbols and notations used throughout this paper. 

\begin{table}[t!]
\centering
\caption{Table of notations}
\label{tab:notation}
\renewcommand{\arraystretch}{1.1}
\begin{tabularx}{\columnwidth}{lX}
\toprule
Symbol & Description \\
\midrule
$G=(V,E)$ & Directed graph with node set $V$ and edge set $E$ \\
$S$ & Seed set of initially positive nodes \\
$X$ & Set of prebunking (intervention) targets \\
$k$ & Intervention budget \\
$p_{uv}$ & Propagation probability from node $u$ to $v$ \\
$q_v$ & Misinformation susceptibility of node $v$ \\
$\varepsilon_v$ & Individual intervention effect on node $v$ \\
$q_v^X$ & Post-intervention susceptibility \\
$ap^+(v;S,X,G)$ & Probability that node $v$ eventually shares misinformation\\
$ap^-(v;S,X,G)$ & Probability that node $v$ eventually shares corrective information\\
$\sigma_G^+(S,X)$ & Expected spread of misinformation \\
$\sigma_G^-(S,X)$ & Expected spread of corrective information \\
$P_{uv}$ & A path from node $u$ to node $v$ in graph $G$ \\
$\MIP_{uv}$ & Maximum Influence Path from $u$ to $v$ \\
$pp(P_{uv})$ & Propagation probability along path $P_{uv}$ \\
$\theta$ & Influence threshold in the MIA model \\
$\MIIA(v)$ & Maximum Influence In-Arborescence of $v$ \\
$\MIOA(v)$ & Maximum Influence Out-Arborescence of $v$ \\
$ap_t^+(v)$ & Probability that $v$ is positive at time $t$ \\
$ap_t^-(v)$ & Probability that $v$ is negative at time $t$ \\
$\pi_t^+(v)$ & Probability that $v$ becomes positive \textit{for the first time} at time $t$ \\
$\pi_t^-(v)$ & Probability that $v$ becomes negative \textit{for the first time} at time $t$ \\
$\beta_t^-(v)$ & Probability that node $v$ does not receive corrective information at time $t+1$ \\
$\beta_t(v)$ & Probability that node $v$ does not receive any information at time $t+1$ \\
\bottomrule
\end{tabularx}
\end{table}

\section{Network Prebunking Problem Under Uncertain Seed Nodes}
\label{sec:NPP-US}
In the network prebunking problem of Eq.~(\ref{eq:NPP}), we assume that the seed node set $S$ is fixed.
This corresponds to the case where misinformation spreads from specific opinion leaders or media accounts on social media.
In reality, however, misinformation is not always disseminated by suce sources; 
it can also originate from ordinary users.
To capture this, we extend the problem to the case where the seed set $S$ is drawn from a probability distribution rather than being fixed.

In graph $G$, let $s(v)$ denote the probability that each node $v \in V$ becomes a seed node.
Then, the probability $s(S)$ that the set of seed nodes is $S$ is given by
\begin{align}
s(S) = \prod_{v \in S} s(v) \prod_{v \not \in S} (1 - s(v))
\end{align}
Using this, the network prebunking problem in graph $G$ with uncertain seed nodes can be formulated as follows:
\begin{align}
    X^* = \argmin_{X \subseteq V} ~\sum_{S \subseteq 2^V} s(S)\, \sigma_G^+(S, X) \quad \text{s.t.} \quad |X|=k.
    \label{eq:NPP-US}
\end{align}

In fact, this problem can be interpreted as the network prebunking problem with a fixed seed node in a new graph.
Consider the new graph $H = (V_H, E_H)$ obtained by adding a dummy node $\alpha$ to graph $G = (V, E)$, where $V_H = V \cup \{\alpha\}$ and $E_H = E \cup \{(\alpha, v) \mid v \in V\}$.
For the newly added edge $(\alpha, v) \in E_H \setminus E$ in graph $H$, let the propagation probability be $p_{\alpha v} = s(v)$.
We set the susceptibility of the dummy node $\alpha$ to $q_\alpha = 1$ and the individual intervention effect to $\varepsilon_\alpha = 0$.
Then, in the graph $H$, the network prebuking problem with seed node $S = \{\alpha\}$ is given by
\begin{align}
    X^* = \argmin_{X \subseteq V} ~~ \sigma_{H}^+(\{\alpha\}, X) \quad \text{s.t.} \quad |X|=k.
\end{align}
This is equivalent to the network prebunking problem under uncertain seed nodes in Eq.~(\ref{eq:NPP-US}).
Therefore, for the network prebunking problem under uncertain seed nodes, we can also apply MIA-NPP directly by introducing a dummy node.

\section{Relation to other problems}
\label{sect:other_problems}

In this section, we discuss the relation of the network prebunking problem to the Influence Minimization (IMIN) problem and the Influence Blocking Maximization (IBM) problem.

\subsection{Relation to IMIN Problem}
In a graph $G$, under the IC model, the expected misinformation spread for a given seed node set $S \subset V$ and a blocking node set $B \subset V \setminus S$ can be written as $\sigma_{G[V \setminus B]}(S) = \sum_{v \in V} ap(v; S, G[V \setminus B])$. 
Therefore, the IMIN problem by node blocking can be formulated as:
\begin{align}
    B^\ast = \argmin_{B \subseteq V \setminus S} ~\sigma_{G[V \setminus B]}(S) \quad \text{s.t.}\quad |B|=k.
    \label{eq:IMIN}
\end{align}

In the network prebunking problem, assume that $q_v = \varepsilon_v = 1$ for all nodes $v \in V$ and that the propagation probabilities for each edge $(u,v) \in E$ are given by $p_{uv}^+ = p_{uv}$ and $p_{uv}^- = 0$.
Under these assumptions, a non-intervened node always shares the misinformation upon receiving it, while an intervened node share neither misinformation nor corrective information, implying that prebunking practically functions as node blocking.
Therefore, the network prebunking problem includes the IMIN problem as a special case.

\subsection{Relation to IBM Problem}
The IBM problem aims to hinder the spread of misinformation originating from a seed set $S_M$ by injecting corrective information from an optimal seed set $S_T \subseteq V \setminus R_\tau(S_M)$ after a delay of $\tau$ steps, where $R_\tau(S_M)$ denotes the set of nodes reachable from $S_M$ within $\tau$ steps.
For given $S_M$ and $S_T$, let $U(S_M, S_T)$ denote the set of nodes that shared misinformation after diffusion process.
Then, the blocked influence of $S_T$ is defined as $\varphi_G(S_M, S_T) = \mathbb{E}[|U(S_M, \emptyset) \setminus U(S_M, S_T)|]$, which represents the expected number of nodes that would have adopted misinformation without $S_T$, but did not do so because of the presence of $S_T$.
The IBM problem can thus be formulated as follows:
\begin{align}
    S_T^\ast = \argmax_{S_T \subseteq V \setminus R_\tau(S_M)} ~\varphi_{G}(S_M, S_T), \quad \text{s.t.}\quad |S_T|=k.
    \label{eq:IBM}
\end{align}
Note that since $U(S_M, S_T)$ is monotonically decreasing with respect to $S_T$, the same problem can be expressed as finding $S_T^\ast$ that minimizes the misinformation spread $\sigma_{G}^M(S_M, S_T) = \mathbb{E}[|U(S_M, S_T)|]$ as follows:
\begin{align}
    S_T^\ast = \argmin_{S_T \subseteq V \setminus R_\tau(S_M)} ~\sigma_{G}^M(S_M, S_T), \quad \text{s.t.}\quad |S_T|=k.
    \label{eq:IBM2}
\end{align}

In the network prebunking problem, suppose that $q_v = \varepsilon_v = 1$ for all nodes $v \in V$.
Under this setting, a non-intervened node always shares the misinformation upon receiving it, while an intervened node is triggered by the arrival of misinformation to immediately share corrective information.
This behavior can be captured by defining a node-specific delay $\tau_v = t_S^+(v)$, where $t_S^+(v)$ is the time at which node $v$ first receives misinformation from the seed set $S$.
Accordingly, the network prebunking problem with $q_v = \varepsilon_v = 1$ is equivalent to the IBM problem with node-specific delays assigned to each node.

\section{Comparison with Greedy Baseline on Small Synthetic Graphs}
\label{appendix:compare_greedy}

To empirically position MIA-NPP against the theoretical greedy benchmark, we conducted additional experiments on small synthetic networks where the greedy algorithm with CELF~\cite{leskovec2007cost} can be executed within a reasonable time. 
Specifically, we evaluated both methods on Erdős-Rényi (ER) and Barabási-Albert (BA) networks with 500 nodes, where the ER network was generated with edge probability $p_{\mathrm{ER}} = 0.05$ and the BA network with attachment parameter $m_{\mathrm{BA}} = 10$. 
The Random baseline, which selects prebunking nodes uniformly at random, was also included for reference.

Figure~\ref{fig:appendix_greedy} compares the relative misinformation spread $y(k)/y(0)$ as a function of the number of prebunked nodes $k$. 
When the number of MC samples $n_{\mathrm{MC}}$ for the greedy algorithm is sufficiently large ($n_{\mathrm{MC}} \ge$ 20,000), its performance is marginally superior; 
however, MIA-NPP with $\theta=0.001$ closely matches it across all $k$ values. 
For smaller $n_{\mathrm{MC}}$, the greedy results become noticeably unstable due to larger estimation errors $\varepsilon_{\mathrm{MC}}$, reflecting the difficulty of distinguishing the small marginal gains between interventions. 

Table~\ref{tab:runtime} summarizes the runtime of the MIA-NPP and the greedy algorithm. 
MIA-NPP completes in roughly 20 seconds, whereas the greedy algorithm requires several thousand seconds even for a small $n_{\mathrm{MC}}$ and exceeds 20,000 seconds for large $n_{\mathrm{MC}}$, demonstrating that MIA-NPP achieves near-greedy effectiveness with approximately two orders of magnitude lower computational cost.

\begin{figure}[t]
    \centering
    \includegraphics[width=\linewidth]{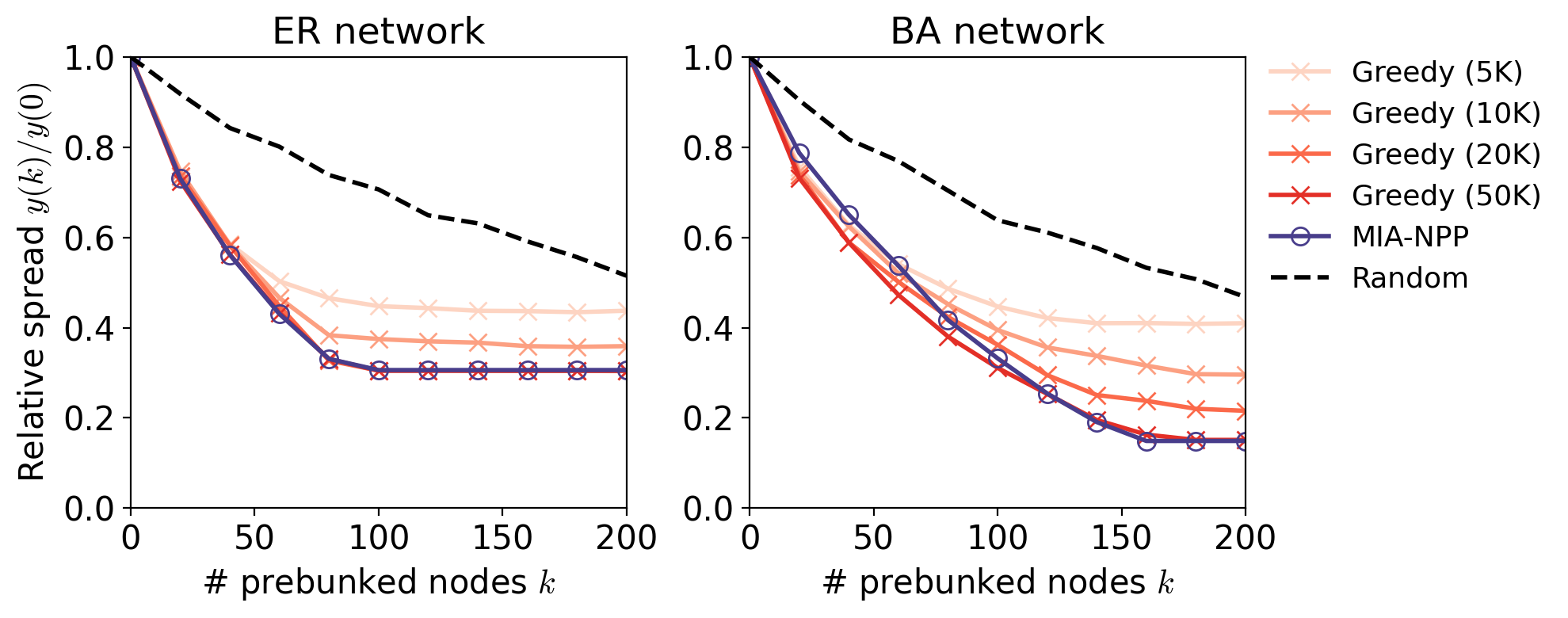}
    \caption{
    Performance comparison of MIA-NPP and the greedy algorithm on ER ($n=500$, $p=0.05$) 
    and BA ($n=500$, $m=10$) networks.
    }
    \label{fig:appendix_greedy}
\end{figure}

\begin{table}[t]
\centering
\caption{Runtime (in seconds) of MIA-NPP and the greedy algorithm with different numbers of MC samples.}
\label{tab:runtime}
\begin{tabular}{lrr}
\toprule
 & ER network & BA network\\
\midrule
MIA-NPP & 20.76 & 17.49 \\
Greedy ($n_{\mathrm{MC}}=$ 5,000)  & 4,060.92  & 5,032.28 \\
Greedy ($n_{\mathrm{MC}}=$ 10,000) & 8,756.90  & 10,954.64 \\
Greedy ($n_{\mathrm{MC}}=$ 20,000) & 17,938.33 & 22,692.15 \\
Greedy ($n_{\mathrm{MC}}=$ 50,000)  & 43,113.94  & 61,982.90 \\
\bottomrule
\end{tabular}
\end{table}

\section{Sensitivity and execution time for influence threshold $\theta$}
\label{sect:sensitivity_analysis}
In this section, we analyze the sensitivity of misinformation suppression effect of MIA-NPP and the execution time with respect to the influence threshold parameter $\theta$.
Under the same experimental conditions as in Section~\ref{sect:eval_perfect}, we vary $\theta$ among $\{0.1, 0.01, 0.001, 0.0001\}$ and report the resulting suppression effect and execution time in Fig.~\ref{fig:sensitivity} and Fig.~\ref{fig:execution_time}, respectively.
As shown in Fig.~\ref{fig:sensitivity}, the value of $\theta$ has almost no effect on the suppression effect of MIA-NPP.
This indicates that the set of intervention targets selected by MIA-NPP remains largely unchanged even when $\theta$ is decreased.
This can be intuitively understood as follows.
The parameter $\theta$ determines how broadly the candidate set $U = U_\theta$ for intervention targets is considerd.
When the threshold is reduced from $\theta_1$ to a smaller value $\theta_2$ (i.e., $\theta_2 < \theta_1$), the candidate set $U_{\theta_2}$ expands, but the newly added nodes $v \in U_{\theta_2} \setminus U_{\theta_1}$ typically have very low probabilities of being reached from any seed node $s$ via their maximum influence paths $\MIP_{sv}$ (i.e., $p_{su} \le \theta_1$).
Therefore, MIA-NPP is also unlikely to select such nodes as intervention targets.

Moreover, Fig.~\ref{fig:execution_time} shows that the execution time of MIA-NPP decreases as the value of $\theta$ increases. 
This is because computing the MIIA and MIOA structures constitutes the primary bottleneck of MIA-NPP, and a larger $\theta$ reduces both the number and size of MIIAs that need to be constructed.
Note that although both MIA-NPP and CMIA-O are based on the MIA model and share the same threshold parameter $\theta$, CMIA-O may select nodes with low (or even no) influence from the misinformation seed as corrective seed nodes, whereas MIA-NPP considers only those nodes that are likely to be influenced by misinformation as candidates for prebunking interventions. 
As a result, under the same $\theta$, the number of MIOAs and MIIAs that must be constructed in MIA-NPP is relatively smaller, and thus, MIA-NPP tends to require less computation time than CMIA-O.

\begin{figure}[t]
    \centering
    \includegraphics[width=1.0\linewidth]{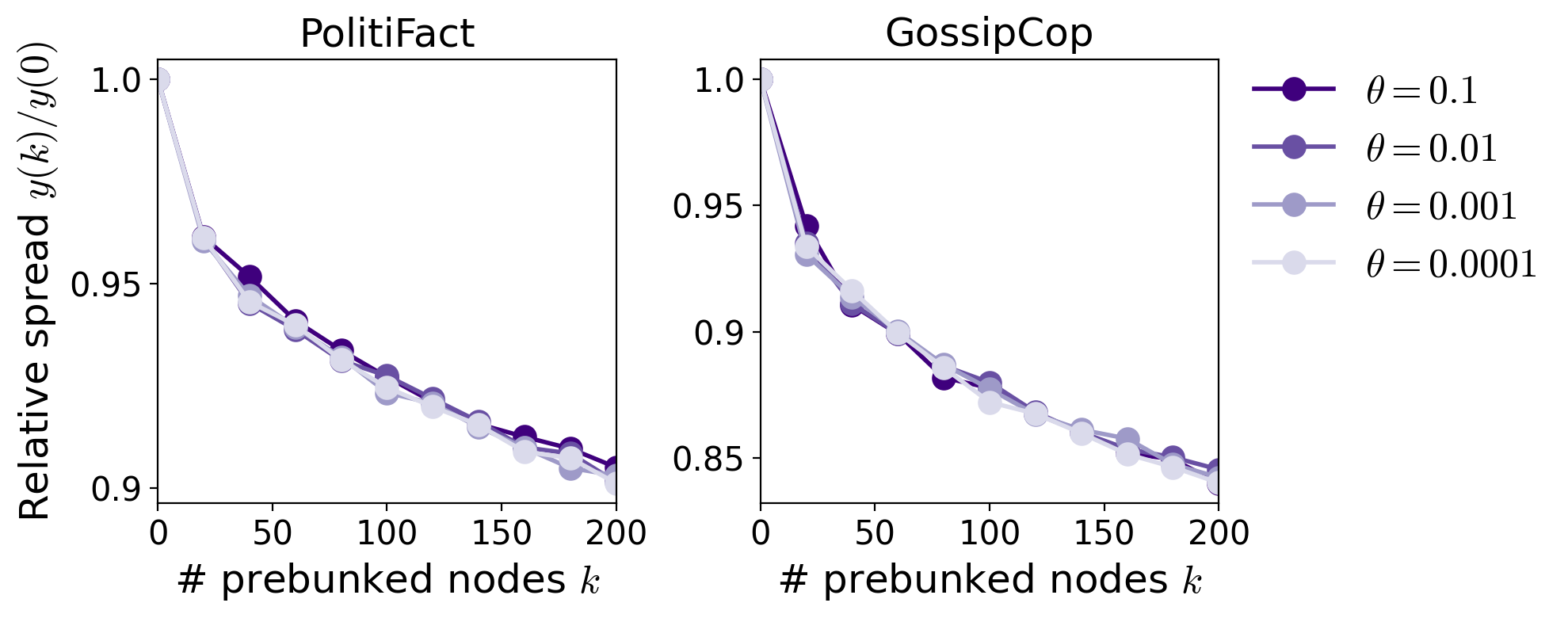}
    \caption{Misinformation suppression effects of MIA-NPP with $\theta \in \{0.1, 0.01, 0.001, 0.0001\}$ on the PolitiFact and GossipCop networks.}
    \label{fig:sensitivity}
\end{figure}

\begin{figure}[t]
    \centering
    \includegraphics[width=1.0\linewidth]{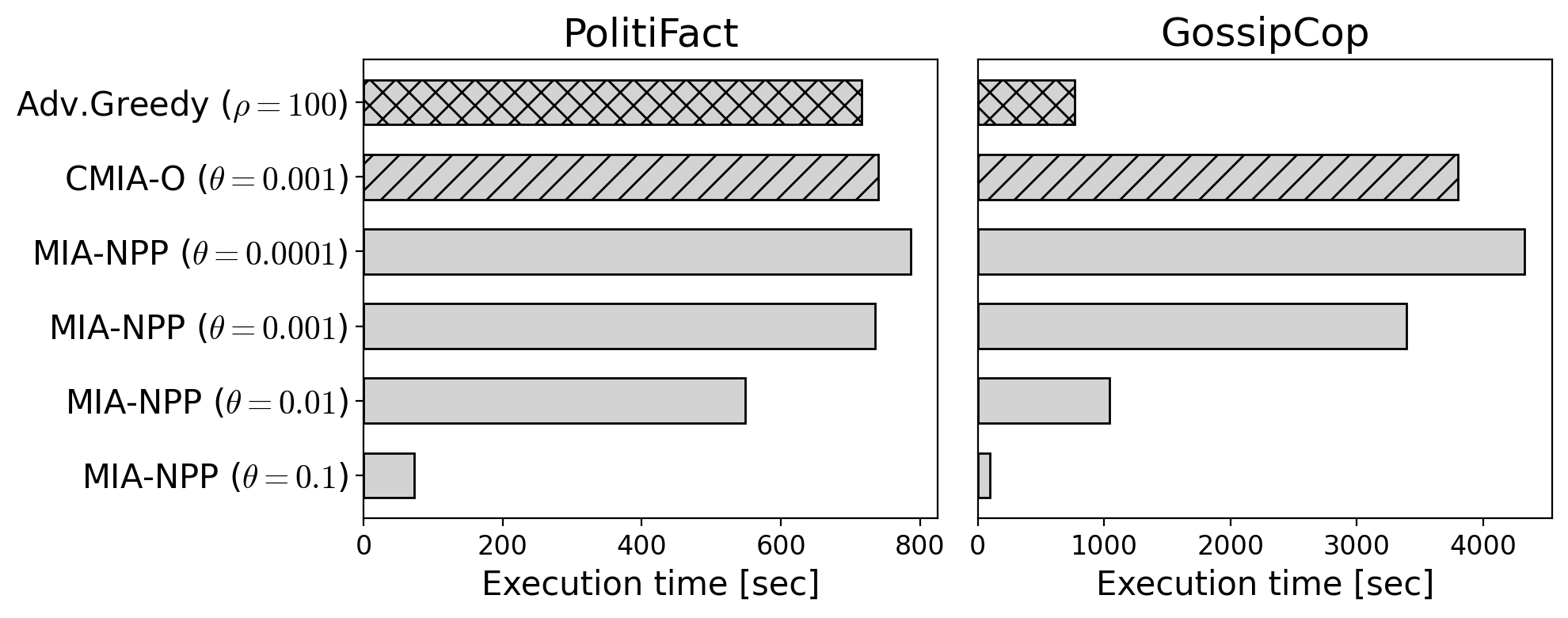}
    \caption{Execution time of MIA-NPP for different $\theta$.}
    \label{fig:execution_time}
\end{figure}

\begin{figure*}[t!]
    \centering
    \includegraphics[width=1.0\linewidth]{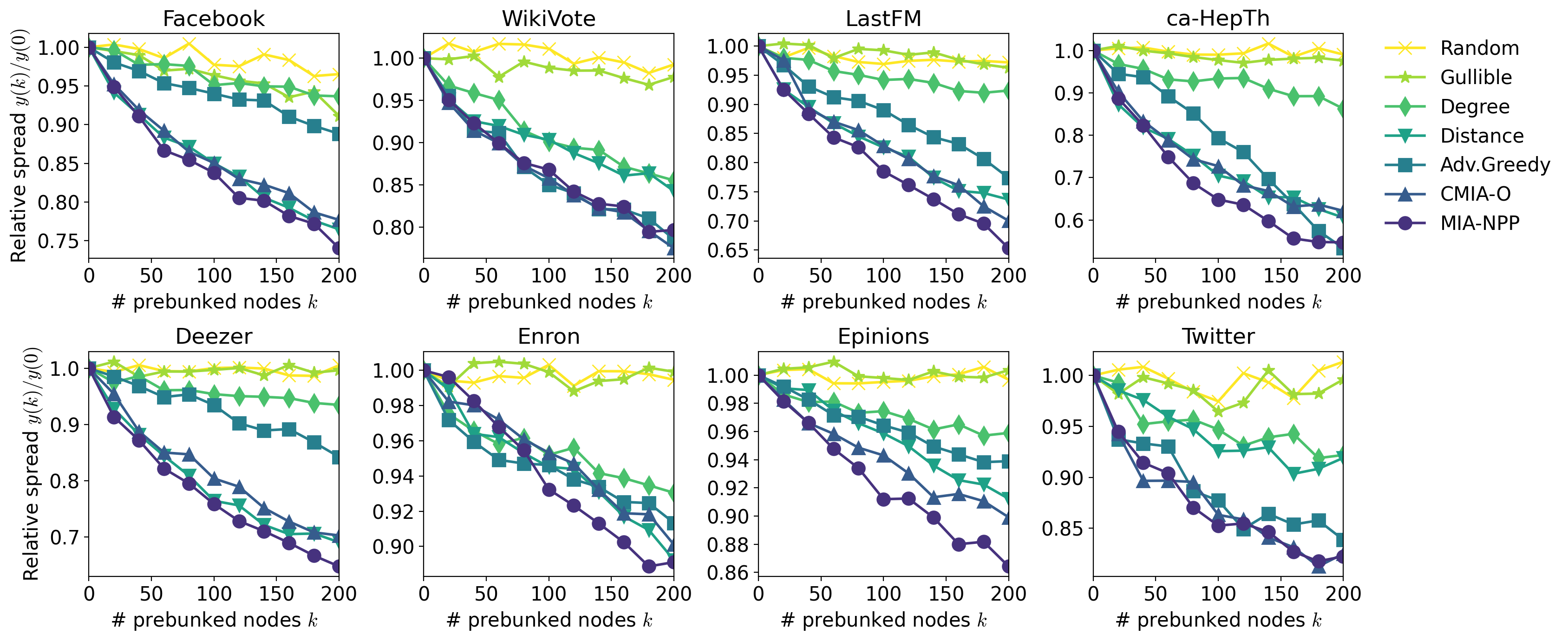}
    \caption{Misinformation suppression effects of each algorithm on various social networks.}
    \label{fig:results_synthetic}
\end{figure*}

\section{Suppression Effects for Other Social Networks}
\label{sect:results_for_other_datasets}

To evaluate whether MIA-NPP remains robust across general social networks beyond the PolitiFact and GossipCop networks, we compared its misinformation suppression effect with that of other algorithms on eight real-world social networks listed in Table~\ref{tab:stat_datasets}.
Since these datasets do not provide misinformation susceptibility parameters $q_v$ for individual nodes, we synthetically assigned $q_v$ values by sampling from a truncated normal distribution $\mathcal{N}_{[0,1]}(0.7, 0.3)$ with mean $\mu_q = 0.7$ and variance $\sigma_q^2 = 0.3$. 
Following the weighted cascade (WC) model~\cite{kempe2003maximizing}, we set the propagation probability $p_{uv}$ for each edge $(u, v)$ as $p_{uv} = 1 / d_v^{\mathrm{in}}$, where $d_v^{\mathrm{in}}$ is the in-degree of node $v$.
The seed nodes were selected by randomly choosing five nodes from among the top 50 nodes with the highest out-degrees in each network. 
This is because, if seed nodes are randomly selected from all nodes, nodes with low out-degree (which constitute the majority in these networks) are more likely to be selected as seed nodes.
Such nodes rarely trigger sufficient diffusion, making it difficult to conduct a proper evaluation.
All other experimental settings follow those described in Section~\ref{sect:eval_perfect}. 

The results are shown in Fig.~\ref{fig:results_synthetic}. 
As the figure illustrates, MIA-NPP consistently outperforms all other algorithms in suppressing the spread of misinformation across all networks.

\begin{table}[h]
\caption{Statistics of datasets.}
\label{tab:stat_datasets}
\centering
 \begin{tabular}{lrrrc}
  \toprule     
  Network
  & \thead{$|V|$} & \thead{$|E|$} & \thead{$d_{\max}^{\mathrm{out}}$} & Type \\
  \midrule       
  Facebook & 4,039 & 88,234 & 1,045 & Undirected \\
  WikiVote & 7,115 & 103,689 & 893 & Directed \\ 
  LastFM & 7,624 & 27,806 & 216 & Undirected \\
  ca-HepTh & 8,638 & 49,633 & 65 & Undirected \\
  Deezer & 28,281 & 92,752 & 172 & Undirected \\
  Enron & 36,692 & 183,831 & 1,383 & Undirected \\
  Epinions & 75,879 & 508,837 & 1,801 & Directed \\
  Twitter & 81,306 & 1,768,149 & 1,205 & Directed \\
  \bottomrule     
 \end{tabular}
\end{table}

\section{Comparison of Suppression Effects: Prebunking vs. Blocking vs. Clarification}
\label{sect:compare_PBC}

In this section, we experimentally compare the effectiveness of prebunking with two representative misinformation mitigation approaches: blocking and clarification.

To evaluate the suppression effect of each approach, we formulate different misinformation minimization problems on the same social network.
For the blocking-based approach, we consider the IMIN problem under the IC model (Eq.~(\ref{eq:IMIN})), and identify a blocking node set $B$ by AdvancedGreedy.
We then compute the relative spread of misinformation $y(k)/y(0)$ when the top-$k$ nodes in $B$ are blocked.
For the clarification-based approach, we consider the IBM problem under the COICM model (Eq.~(\ref{eq:IBM})), and select a seed set $S_T$ for spreading corrective information by CMIA-O.
We evaluate its effectiveness by computing the relative spread $y(k)/y(0)$ when corrective information is launched from the top-$k$ nodes in $S_T$ after a delay of $\tau \in \{0,1\}$.
The delay parameter $\tau$ represents the number of steps after the misinformation is released before the clarification campaign can begin, which implies that no node within $\tau$ hops of any misinformation seed node can be selected as a clarification seed.
In particular, when $\tau = 1$, the seed set $S_T$ must exclude both the misinformation seed nodes $S_M$ and their out-neighbors $\Nout_{S_M}$.
For the prebunking approach, we evaluate the performance in the same way as in Section~\ref{sect:eval_perfect}, but vary the mean of the intervention effect distribution $\mu_\varepsilon \in {0.2, 0.5, 1.0}$ to control the strength of the intervention.
Note that since each approach assumes a different diffusion model, direct comparison of the absolute misinformation spread $y(k)$ is not meaningful.

Figure~\ref{fig:results_compare_PBC} presents the comparison results of misinformation suppression effects across the three approaches.
First, the blocking (B) approach, which can completely cut off the inflow of misinformation to the selected nodes, achieves the most effective suppression compared to the other methods.
However, despite its effectiveness, blocking is highly intrusive, as it entails destructive modifications to the network structure.
The clarification (C) approach achieves suppression performance comparable to blocking when there is no delay in injecting corrective information.
This is because, in the clarification setting, the selected seed nodes always share corrective information and never share misinformation, thereby completely blocking the inflow of misinformation to these nodes.
Effectively, this functions similarly to node blocking.
Thus, it is intuitive that clarification without delay exhibits similar suppression effects to blocking.
However, when a delay is introduced, the suppression effect of clarification significantly deteriorates due to the limited reach of corrective information.
This tendency is consistent with findings in prior studies~\cite{budak2011limiting,wen2014shut}.
The suppression effect of the prebunking (P) approach improves as $\mu_\varepsilon$ increases.
Nevertheless, its effectiveness remains lower than that of blocking or clarification without delay.
This is because, unlike the other approaches, prebunking only reduces the probability of misinformation sharing and does not guarantee that intervened nodes will refrain from sharing misinformation.

\begin{figure}[t]
    \centering
    \includegraphics[width=1.0\linewidth]{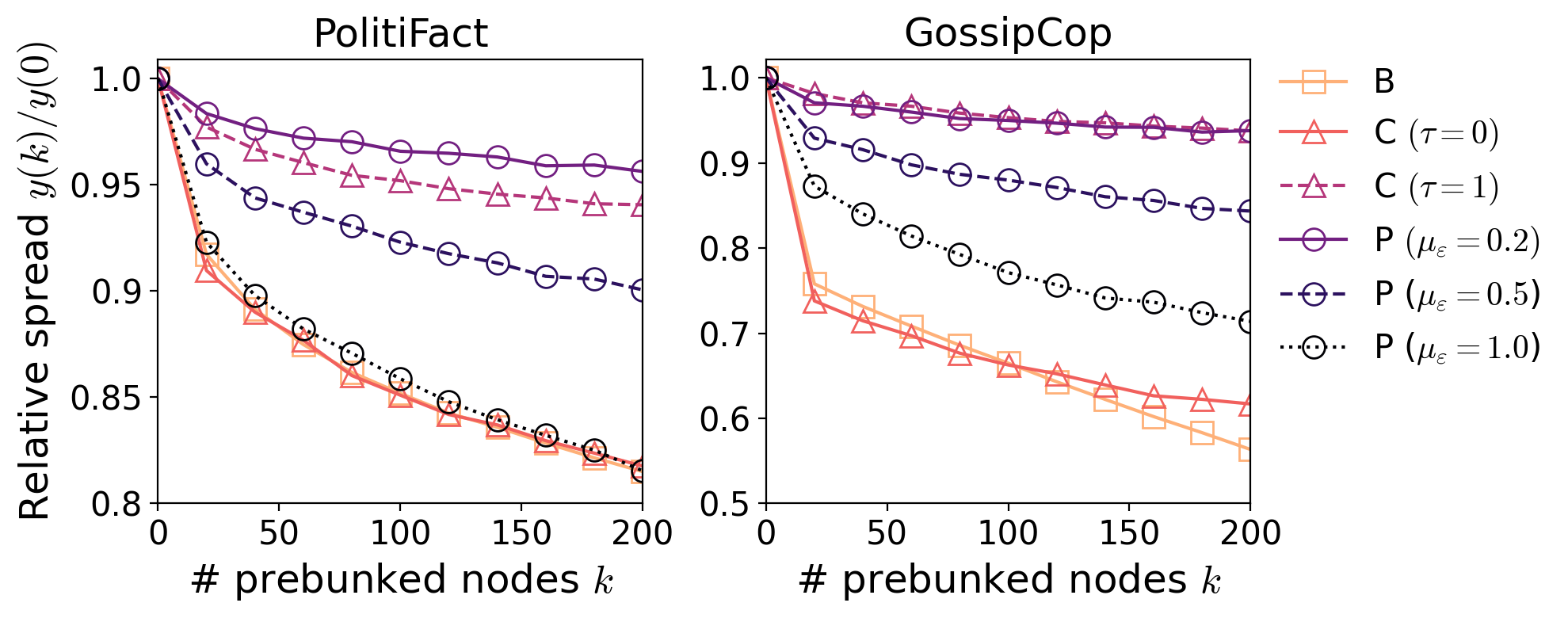}
    \caption{Comparison of misinformation suppression effects of Blocking (B), Clarification (C), and Prebunking (P) approaches on the PolitiFact and GossipCop networks.}
    \label{fig:results_compare_PBC}
\end{figure}

\end{document}